\documentclass[oneside,fleqn,reqno,english]{amsart}
\usepackage[T1]{fontenc}
\usepackage[latin9]{inputenc}
\usepackage{color}
\usepackage{babel}
\usepackage{verbatim}
\usepackage{calc}
\usepackage{amstext}
\usepackage{amsthm, amssymb}
\usepackage{graphicx}
\PassOptionsToPackage{normalem}{ulem}
\usepackage{ulem}
\usepackage[unicode=true,
 bookmarks=true,bookmarksnumbered=false,bookmarksopen=false,
 breaklinks=false,pdfborder={0 0 1},backref=section,colorlinks=true]
 {hyperref}
\hypersetup{
 plainpages=false,bookmarksopen=true,citecolor=teal,linkcolor=teal,urlcolor=teal}

\makeatletter
\numberwithin{equation}{section}
\numberwithin{figure}{section}
\theoremstyle{plain}
\newtheorem{thm}{\protect\theoremname}[section]
\theoremstyle{plain}
\newtheorem{fact}[thm]{\protect\factname}
\theoremstyle{plain}
\newtheorem{lem}[thm]{\protect\lemmaname}
\theoremstyle{remark}
\newtheorem{claim}[thm]{\protect\claimname}
\theoremstyle{remark}
\newtheorem{rem}[thm]{\protect\remarkname}
\theoremstyle{plain}
\newtheorem{prop}[thm]{\protect\propositionname}
\theoremstyle{plain}
\newtheorem{cor}[thm]{\protect\corollaryname}

\usepackage{doi}
\usepackage{xcolor}
\usepackage{doi}
\usepackage{enumitem}
\usepackage{hyperref}

\hypersetup{
plainpages = false,
bookmarksopen = true,
colorlinks = true,
citecolor = teal,
linkcolor = teal,
urlcolor  = teal,
}

\newcommand\thankssymb[1]{\textsuperscript{\@fnsymbol{#1}}}
\setlist[itemize,1]{leftmargin=\dimexpr 26pt-.1in}
\setlist[enumerate,1]{leftmargin=\dimexpr 26pt-.1in}
\setenumerate{label=\rm (\roman*)}

\numberwithin{equation}{section}

\def\@seccntformat#1{%
  \protect\textup{%
    \protect\@secnumfont
    \expandafter\protect\csname format#1\endcsname 
    \csname the#1\endcsname
    \protect\@secnumpunct
  }%
}

\newcommand{\SparseDotfill}{\leavevmode \leaders \hb@xt@ .7em{\hss .\hss }\hfill \kern \z@}
  
\usepackage{etoolbox}

\patchcmd{\@tocline}
  {\hfil}
  {\leaders\hbox{\,.\,}\hfil}
  {}{}

\newtheorem{theorem}{Theorem}[section]

\usepackage{template}

\newcommand{\cube}[2]{\mathrm{cube}_{#1,#2}}

\makeatother

\providecommand{\claimname}{Claim}
\providecommand{\corollaryname}{Corollary}
\providecommand{\factname}{Fact}
\providecommand{\lemmaname}{Lemma}
\providecommand{\propositionname}{Proposition}
\providecommand{\remarkname}{Remark}
\providecommand{\theoremname}{Theorem}

\begin{document}
\title{An Optimal ``It Ain't Over Till It's Over'' Theorem}
\author{Ronen Eldan\thankssymb{1}}
\thanks{\thankssymb{1}Microsoft Research. Partially supported by NSF grant
CCF-1900460. \texttt{E-mail:roneneldan@microsoft.com}}
\author{Avi Wigderson\thankssymb{2}}
\thanks{\thankssymb{2}Institute for Advanced Study, Princeton, NJ 08540.
Supported by NSF grant CCF-1900460. \texttt{E-mail:avi@ias.edu}}
\author{Pei Wu\thankssymb{3}}
\thanks{\thankssymb{3}Institute for Advanced Study, Princeton, NJ 08540.
Supported by NSF grant CCF-1900460. \texttt{E-mail:pwu@ias.edu}}
\begin{abstract}
We study the probability of Boolean functions with small max influence
to become constant under random restrictions. Let $f$ be a Boolean
function such that the variance of $f$ is $\Omega(1)$ and all its
individual influences are bounded by $\tau$. We show that when restricting
all but a $\rho=\tilde{\Omega}((\log1/\tau)^{-1})$ fraction of
the coordinates, the restricted function remains nonconstant with
overwhelming probability. This bound is essentially optimal, as witnessed
by the tribes function $\TR=\AND_{n/C\log n}\circ\OR_{C\log n}$.

We extend it to an anti-concentration result, showing that the restricted
function has nontrivial variance with probability $1-o(1)$. This
gives a sharp version of the ``it ain't over till it's over'' theorem
due to Mossel, O'Donnell, and Oleszkiewicz. Our proof is discrete,
and avoids the use of the invariance principle.

We also show two consequences of our above result:
(i) As a corollary, we prove that for a uniformly random input $x$, the
block sensitivity of $f$ at $x$ is $\tilde{\Omega}(\log1/\tau)$
with probability $1-o(1)$. This should be compared with the implication
of Kahn, Kalai and Linial's result, which implies that the average block
sensitivity of $f$ is $\Omega(\log1/\tau)$. (ii) Combining our proof with
a well-known result due to O'Donnell, Saks, Schramm and Servedio, 
one can also conclude that: 
Restricting all but a $\rho=\tilde\Omega(1/\sqrt{\log (1/\tau) })$ fraction of the coordinates of a monotone
function $f$, then the restricted function
has decision tree complexity $\Omega(\tau^{-\Theta(\rho)})$ with probability $\Omega(1)$. 

\end{abstract}

\maketitle
\thispagestyle{empty} \newpage{}

\section{Introduction}

For any Boolean function $f:\{-1,1\}^{n}\to\{0,1\}$, the individual
\emph{influence} of the $i$th coordinate is the probability of flipping
the value of $f$ by flipping $x_{i}$ on a random input $x$. Let $x \oplus (-1)^{e_i}$
denote the string obtained by flipping the $i$th coordinate of $x$, then
\[
\INF_{i}(f):=\PP_{x\in\DC}[f(x)\not= f(x\oplus (-1)^{e_i})].
\]
In this paper, we study Boolean functions with small influences, hence
functions satisfying 
\[
\INF_{\infty}(f):=\max_{i\in[n]}\INF_{i}(f)=o(1).\footnotemark
\]
\footnotetext{For the rest of the paper, we consider the function
$f$ as a family of functions. Thus here by $o(\cdot)$, we
mean ``as $n$ goes to infinity.'' The bound $o(1)$ on the influences is
worse than needed. For illustration, this is good enough as many examples
we are interested in this paper satisfy that their influences 
are $o(1)$.}

Let $\cR_{p}$ denote a $p$-\emph{random} \emph{restriction}, namely,
a randomly-chosen subcube where for each coordinate, one flips a coin
and, with probability $p$ one fixes the value of the coordinate
to $-1$ or $1$ (with equal probabilities) and, with probability
$1-p$, the coordinate is left undetermined (alive). Then $f|_{\cR_{p}}$
is a random sub-function given by restricting $f$ to the subcube.

In this paper, we study Boolean functions with small influences under
random restrictions. Our main goal is to prove a lower bound for the
probability of the function to remain nonconstant under the restriction.
We prove the following near-optimal result: 
\begin{thm}[A simplified version of Theorem~\ref{thm:main}]
\label{thm:mainintro} Given $f:\{-1,1\}^{n}\to\{0,1\}$ such that
the variance of $f$ is $\Omega(1),$ and $\tau:=\INF_{\infty}(f)=o(1).$
Let $\cR_{1-\rho}$ be a random restriction where 
\[
\rho=\Omega\left(\frac{\log\log(1/\tau)}{\log(1/\tau)}\right).
\]
Then for any $p\ge \mINF(f)^{\Theta(\rho)}$,
\[
\PP[\Var[f|_{\cR_{1-\rho}}]\le p^{\tilde{\Theta}\left(\frac{1}{\rho}\right)}]\le p.
\]
\end{thm}

The bound on the variance is near-optimal by the majority function.
In particular, if $f$ is the majority function, then 
\[
\PP[\Var[f|_{\cR_{1-\rho}}]\le p^{\Theta\left(\frac{1}{\rho}\right)}]\le p.
\]
Furthermore, our bound on $\rho$ is optimal up to a $\log\log$ factor.
Because randomly restricting the tribes function with $\rho=O(1/\log(1/\tau)),$
we get a constant function with probability $\Omega(1)$. Previously,
Mossel et al. proved a similar result for $\rho=\Omega(1/\sqrt{\log(1/\tau)})$
using completely different techniques~\cite{mossel2010noise}. Prior to
Mossel et al.'s work, the related conjecture, with a very suggestive name ``it
ain't over till it's over'' conjecture, was proposed by Kalai and
Friedgut in studying social indeterminacy~\cite{kalai02arrow,kalai04social}.
It implies a quantitative version
of the Arrow's Theorem. We refer the interested readers to~\cite{mossel2010noise} for
more discussions.

Next, we discuss a corollary of this theorem to block sensitivity
of functions with small influences. The \emph{sensitivity }of an input
$x$ with respect to Boolean function $f$, denoted $\sen_{f}(x):=\sum_{i\in[n]}[f(x)\not=f(x\oplus (-1)^{e_i})]$, is the number
of the Hamming neighbors of $x$ which have a different function value.
An inequality by Kahn, Kalai and Linial \cite{KKL88} asserts that
\[
\EE_{x}[s_{f}(x)]=\Omega\left(\log\frac{\Var[f]}{\INF_{\infty}(f)}\right),
\]
which naturally leads to the question of whether it is also true that
$s_{f}(x)=\Omega\left(\log\frac{1}{\INF_{\infty}(f)}\right)$
for a \emph{typical} point $x$. This is clearly not the case, as
witnessed by the majority function. However, a corollary to our theorem
is that such an estimate does indeed hold true for most points $x$,
if sensitivity is replaced by the related notion of \emph{block sensitivity}.

The block sensitivity of an input $x$ with respect to function $f$,
denoted $\bs_{f}(x)$ is the maximum number of disjoint sets $S_{1},S_{2},\ldots,S_{m}\subseteq[n]$,
such that for $i\in[m]$, one has $f(x)\not=f(x\oplus(-1)^{\1_{S_{i}}})$,
by $x\oplus(-1)^{\1_{S_{i}}}$ we mean flipping the sign of variables
in $S_{i}$. Clearly,\footnote{By taking singleton sets above.} one has $\bs_{f}(x)\geq s_{f}(x)$ for all $f,x$.
Our second result shows that for functions with small influences,
the block sensitivity is large on almost all points $x$: 
\begin{thm}
\label{thm:blockintro} For any function $f:\{-1,1\}^{n}\to\{0,1\}$
such that its variance is $\Omega(1),$ and $\tau:=\INF_{\infty}(f)=o(1).$
Then 
\begin{align*}
 & \PP_{x}[\mathbf{\bs}_{f}(x)\ge\tilde{\Omega}(\log1/\tau)]=1-o(1).
\end{align*}
\end{thm}

Finally, if the function $f$ is monotone in addition to having small influences,
our analysis to Theorem~\ref{thm:mainintro} implies an upper bound on the influences
of $f$ under random restrictions. In the work due to O'Donnell et al.~\cite{osss2005dt-maxinf},
it is proved that every shallow decision tree must have an influential
variable. Combining these facts, one can also conclude that, for 
monotone function $f$, the restricted function will have large decision
tree complexity.  In particular, let $\DT(f)$ denote the decision tree complexity of $f$.
Then,
\begin{theorem}
\label{thm:dt-intro}
For any monotone function $f:\DC\to\{0,1\}$ with $\Omega(1)$ variance, 
and $\tau=\mINF(f)=o(1)$. Then for any $\rho=\tilde\Omega(\sqrt{1/\log (1/\tau) })$,
\begin{equation*}
    \PP[\DT(f|_{\cR_{1-\rho}})=\tau^{-\Theta(\rho)}]\ge \frac{1}{2}.
\end{equation*}
\end{theorem}
The above theorem is, in a sense, a reverse statement to the H{\aa}stad switching
lemma, which states that applying the $(1-O(1/\log n))$-random restriction  
to any polynomial-size DNF/CNF (or in general any $\mathrm{AC}^0$ circuits), 
one gets a shallow decision tree with high probability. Our result, on the contrary,
states that random restrictions with
alive probability $\tilde\Omega(1/\log(\mINF(f)))$
cannot simplify $f$ to a too shallow decision tree for monotone functions $f$ with low influences. 

\subsection*{Context and related works}

The notion of influences studied in this paper is first introduced
by Ben-Or and Linial~\cite{BL85} in the context of \emph{collective
coin flipping.} It coincides with the ``Banzaf index'' studied in
game theory. The class of Boolean functions with small influences
have been widely studied. There are several motivations to study such
functions. First, they arise naturally in social choice theory~\cite{kalai02arrow,kalai04social}.
For example, in a voting system of two candidates and $n$ voters,
each bit $x_{i}$ represents the individual preference of each voter
between the two candidates. When aggregating the social preference,
it is natural to use a function $f$ where the potential of any given
individual to determine the final outcome is limited. Second, from
an algorithmic perspective, suppose that we have access to the input
via a limited number of queries. Then, it is natural to query a variable
when its individual influence is large. In many cases, such variables
can be found iteratively and this process leads to a good approximation
of $f$ with a small number of queries. This observation has been
applied in different settings~\cite{Friedgut98,AA14}. In computational
complexity, to distinguish the dictatorship function v.s.
functions with small individual influences is a key component of proving
optimal NP-hardness for approximations~\cite{BGS95,has96,has97,KGMO07UG}.
From an analytic perspective, it has been observed that functions
with small influences exhibit improved concentration inequalities
(e.g., \cite{Talagrand94}) and often tend to exhibit Gaussian-like
behavior~\cite{mossel2010noise}.

Applying random restrictions and studying the properties of the restricted
functions has been widely studied and has led to breakthroughs in
a variety of areas. For example, it is the key idea of the exponential
lower bounds in circuit complexity~\cite{hastad1987} and the dramatic
improvements of the sunflower lemma in combinatorics~\cite{ALWZ21sunflower}.

The problem of determining whether a function with small influences
becomes constant under random restrictions has attracted some attention
in the context of hardness amplification within NP for circuits~\cite{ODONNELL04NPHardness}.

A sub-optimal version of Theorem \ref{thm:mainintro} follows
from the ``It ain't over till it's over'' theorem proven by Mossel,
O'Donnell, and Oleszkiewicz in \cite{mossel2010noise}. Their approach
uses the \emph{invariance principle,} which at a high level asserts
that when feeding a ``smooth''\footnote{By ``smooth'' here, we mean $f$
has low degree.  With additional work, the invariance principle applies
to $f$ that has its Fourier mass concentrated in low degrees.} 
function $f$ with independent random
inputs $X_{1},X_{2},\ldots,X_{n}$ from a product space such that
each $X_{i}$ has zero mean, unit second moment and bounded third
moment, then the output distribution is ``invariant'' to the actual
distribution of the inputs. This approach usually studies a related
problem, then translates the result of the related problem to the
Boolean cube. This translation suffers from two drawbacks. First,
it obscures what is actually happening in Boolean cube. Second, the
requirement of $f$ being ``smooth'' normally requires additional
technical treatment, and becomes the main obstacle for obtaining an
optimal result.

\subsection*{Our approach}

Our approach relies on a control-theory point of view to the problem combined with ideas from ``pathwise-analysis,'' using arguments which are somewhat inspired by~\cite{EG18talagrand-ineq}. We assume that the coordinates are revealed in a random order and are randomly assigned values $\pm1$ one by one. For each coordinate being revealed,
we assume that with probability $\rho$ a player gets an opportunity
to ``override'' the value that has been assigned to that coordinate.\footnotemark
If the player has the capability of deciding, with high probability,
the value of the function, this implies that restricting all but a
$\rho$-fraction of coordinates leaves the restricted function nonconstant.
\footnotetext{We note that Lichtenstein-Linial-Saks~\cite{LLS89control} also study a
control theoretic problem, which in the surface may seem similar. The main difference
is that in their model the player picks which coordinates to influence,
whereas in our model these coordinates are picked randomly, as we
will see momentarily. The two models differ drastically in their nature.}

To this end, assume that $X(t)\in\{-1,0,1\}^{n}$ is the process where
at each step another coordinate is being revealed, where coordinates
whose value is not determined are set to $0$. We view our Boolean
function $f$ as a function $f:\RR^{n}\to\RR$ by considering its
multilinear extension. If the player does not override any coordinate,
then the process $M(t):=f(X(t))$ is a martingale (where, for $x\in\{-1,0,1\}^{n}$,
the expression $f(x)$ is the value of $f$ when taking expectation
over coordinates whose value is set to $0$). The player's ability to override coordinates effectively allows the
player to add a drift to $M(t)$, where the player's goal is to end
up with $M(n)$ being equal to $1$ (or, by the same argument, to
$0$ simply by replacing $f$ with $1-f$).

At this point, let us assume for simplicity that the increments of the martingale 
$M(t)$ have a fixed step size $\eta$ (in other words, assume that it is actually 
a random walk up to the time when it hits $\{0,1\}$). Moreover assume that $M(0)$ 
is bounded away from $\{0,1\}$. Suppose that the player has probability $\rho$ to 
override each step and is trying to force the process to end up at the value $1$ 
by overriding the increment with the value $+\eta$. In this case over $t$ steps 
the process accumulates a variance of $\eta^{2}t$ and a drift of size $t\rho\eta$. Since the
process eventually moves a distance of $\Omega(1)$, and hence accumulates a variance 
of constant order, we have $t\asymp\eta^{-2}$. It follows that in order for the 
effect of the drift to be more significant than that of the variance, one arrives 
at the condition $\rho\gg\eta$. In other words, the process
can be efficiently controlled (meaning that the player gets to determine its endpoint) 
as long as the step size is at most $\rho$. In fact, we will see that this heuristic is only correct when
$M(t)$ is not close to the edges, which will create an additional
technical complication.

The step size of the process is, in turn, is controlled by the $\ell_{\infty}$
norm of the first-order Fourier coefficients of restrictions of the
function $f$, or equivalently by the quantity $\max_{i}|\partial_{i}f(X(t))|$,
where $i$ is over the coordinates not fixed at time $t$.
We need to show that this quantity remains small along the process,
which is where the fact that the initial influences are small will be used.

The control of the first-order Fourier coefficients relies on a new
\emph{hypercontractive} \emph{inequality} for random restrictions.
We consider random restrictions $\cR{}_{p},\cR_{q}$, where $0\le p\le q\le1$
are the probabilities of a variable being fixed, and show that for
any multilinear function $f$ and any $0\le\epsilon\le q-p,$ 
\begin{equation}
(\EE[\mu(f|_{\cR_{p}})^{2+\epsilon}])^{\frac{1}{2+\epsilon}}\le(\EE[\mu(f|_{\cR_{q}})^{2}])^{\frac{1}{2}},
\end{equation}
where we use $\mu(f)$ to denote the expected value of $f$ over
the uniform measure on $\DC$.

The hypercontractive inequality will allow us to control the evolution
of the first-order coefficients under the original (namely, the uncontrolled)
martingale. However, we need to control those coefficients under the
``controlled'' process (where the player gets to override some coordinates).
This can be solved by assuming that the strategy taken by the player
tries to mimic yet another process obtained by \emph{conditioning} the original
martingale $M(t)$ to end up at the value $1$ ($0$, respectively).
This amounts to a change of measure over the space of paths of $X(t)$
which gives tractable formulas for the corresponding change of measure
of a single step. Equivalently, this is the strategy which ensures
ending up at the desired value under a change of measure which has
the minimal possible relative entropy to the uncontrolled process.

Finally, we explain how to strengthen the above result to give 
a quantitative bound on the variance of the restricted
function. We analyze the Kullback-Leibler divergence between, roughly
speaking, the string $Y(n)$ generated by the ``controlled'' process
given the restrictions $\cR_{1-\rho}$ determined by those coordinates
that is not controlled by the player, and a uniformly random string
$X\in\mopo^{n}$. With the Fourier-analytic tool of Level-1 inequality, one can show
that the expected KL-divergence over the random restrictions is about
$\tilde{O}(1/\rho)$. Somewhat surprisingly, the KL-divergence is,
in addition to being small in expectation, highly concentrated. Recall
that $Y(n)$ is sampled from $f^{-1}(1)$. All these imply that $\mu(f|_{\cR_{1-\rho}})\ge\exp(-\tilde{O}(1/\rho))$
with high probability. The variance bound then follows 
once we put together with the other direction that $\mu(f|_{\cR_{1-\rho}})\le1-\exp(-\tilde{O}(1/\rho))$
by replacing $f$ with $1-f$. 

\subsection*{Organization}
We present the necessary preliminaries in Section~\ref{sec:Preliminaries}.
Then in Section~\ref{sec:control}, we carefully define the uncontrolled and
controlled process discussed in the introduction and we study the properties
of these random processes. With this tool at our disposal, we prove our main
result Theorem~\ref{thm:mainintro} in Section~\ref{sec:appl}. Then we explain
the applications of this result to the block sensitivity and decision tree 
complexity. We leave the technical analysis to the final section, 
Section~\ref{sec:hypercontractive}, that the Fourier coefficients
of the first order remain small under random restrictions .

\section{Preliminaries\label{sec:Preliminaries}}

\subsection*{General}

We adopt the the shorthand notation $[n]$ for the set $\{1,2,\ldots,n\}.$
For a string $x\in\{-1,1\}^{n}$ and a set $S\subseteq\{1,2,\ldots,n\},$
we let $x|_{S}$ denote the restriction of $x$ to the indices in
$S.$ In other words, $x|_{S}=x_{i_{1}}x_{i_{2}}\ldots x_{i_{|S|}},$
where $i_{1}<i_{2}<\cdots<i_{|S|}$ are the elements of $S.$ Analogously,
for any function $f:\Omega\to\RR$ over an arbitrary domain $\Omega$.
Let $A\subseteq\Omega$, we adopt the notation $f|_{A}$ for the sub-function
of $f$ over the domain $A$. Namely, $f|_{A}(x)=f(x)$ for $x\in A.$
Given a set $S$, when the universe $U$ is clear from the context
we use $\bar{S}:=U\setminus S$ to denote the complement of $S$.
The \emph{characteristic function} of a set $S$ is given by 
\[
\1_{S}(i)=\begin{cases}
1 & \text{if }i\in S,\\
0 & \text{otherwise.}
\end{cases}
\]
For a permutation $\pi:U\to U$. Let $\pi S$ be the permuted set
of $S$, i.e., $\pi S=\{\pi(i):i\in S\}.$

The primary interest of this paper is Boolean functions $f:\{-1,1\}^{n}\to\{0,1\}.$
Note that we use $-1,1$ to denote ``true'' and ``false'' on the
domain of $f$, respectively. For example, the logic AND function
and the logic OR function are defined as below, 
\begin{align*}
\bigwedge_{i=1}^{n}x_{i}=\begin{cases}
1 & x_{i}=-1\,\,\forall i\in[n],\\
0 & \text{otherwise},
\end{cases}\qquad & \bigvee_{i=1}^{n}x_{i}=\begin{cases}
0 & x_{i}=1\,\,\forall i\in[n],\\
1 & \text{otherwise}.
\end{cases}
\end{align*}
We abuse the notation $x\oplus y$ to denote the entrywise XOR function
for $x,y\in\DC$. Thus $(x\oplus y)_{i}=x_{i}\cdot y_{i}.$ For any
univariate function $h:\RR\to\RR$ and $x\in\RR^{n},$ the application
of $h$ to $x$ means entrywise application, i.e., $h(x)$ denotes
the vector such that $(h(x))_{i}=h(x_{i})$. For any $f:\{-1,1\}^{n}\to\RR,$
$S\subseteq\{1,2,\ldots,n\}$ and $y\in\{-1,1\}^{n}$, let $\cube{S}{y}:=\{x\in\DC:x|_{S}=y|_{S}\}$
be the subcube of $\DC.$ We abbreviate $f|_{(S,y)}=f|_{\cube{S}{y}}.$
The same definition $f|_{(S,y)}$ extends to $y\in\RR^{T}$ for any
$T\supseteq S$ such that $y|_{S}\subseteq\{-1,1\}^{S}$. A random
$p$-restriction $\cR_{p}$ is a random tuple $(S,y)$ such that for each
$i\in[n],$ $i\in S$ with independent probability $p$, and $y$ is a uniformly
random element from $\{-1,1\}^{n}.$ 

For a logical condition $C,$ we use the Iverson bracket 
\[
\II\{C\}=\begin{cases}
1 & \text{if \ensuremath{C} holds,}\\
0 & \text{otherwise.}
\end{cases}
\]
Denote $|x|$ the length of $x$ for any vector $x\in\RR^{n}$, i.e.,
\[
|x|=\left(\sum_{i=1}^{n}x_{i}^{2}\right)^{1/2}.
\]
For two vectors $x,y\in\RR^{n},$ we adopt the following inner product
\[
\langle x,y\rangle=\sum_{i\in[n]}x_{i}y_{i}.
\]
The set $\{e_{1},e_{2},\ldots,e_{n}\}$ forms a standard basis, where
$e_{i}$ denotes the vector whose only nonzero coordinate $i$ is
1.

Given some discrete space $\Omega$ and a probability measure $\gamma$ over $\Omega$.
If the random variable $X$ is drawn from $\gamma$, we denote it
by $X\sim\gamma$. For any function $f:\Omega\to\RR$, we often abbreviate
the expectation of $f$ over $\gamma$ as $\gamma(f),$ namely,
\[
\gamma(f):=\sum_{x\in\Omega}f(x)\gamma(x).
\]
We let $\ln x$ and $\log x$ stand for the natural logarithm of $x$
and the logarithm of $x$ to base $2,$ respectively. For any distribution
$\gamma$ over some discrete space $\Omega$, the entropy function
\[
H(\gamma)=\EE_{x\in\Omega}\gamma(x)\log\frac{1}{\gamma(x)}.
\]
When $\Omega$ contains only two elements, we can think of the binary
entropy function $H\colon[0,1]\to[0,1]$ as given by
\[
H(x)=x\log\frac{1}{x}+(1-x)\log\frac{1}{1-x}.
\]
Basic calculus reveals that for $x\in[-1,1],$
\begin{equation}
1-H(x)\leq4\left(x-\frac{1}{2}\right)^{2}.\label{eq:entropy-upper-bound}
\end{equation}
Recall that the Kullback-Leibler divergence (KL-divergence) between
two distributions $\mu_{0},\mu_{1}$ over $\Omega$ is defined by
the following formula
\[
\KL{\mu_{0}}{\mu_{1}}=\sum_{x\in\Omega}\mu_{0}(x)\log\frac{\mu_{0}(x)}{\mu_{1}(x)}.
\]
The KL-divergence is convex. In particular, let $\mu_{0},\mu_{1}$,
$\gamma_{0},\gamma_{1}$ be distributions over the same space. Then
for any $\lambda\in[0,1],$
\begin{align*}
   & \KL{\lambda\mu_{0}+(1-\lambda\mu_{1})}{\lambda\gamma_{0}+(1-\lambda\gamma_{1}}\\
   & \qquad\qquad \le\lambda\KL{\mu_{0}}{\gamma_{0}}+(1-\lambda)\KL{\mu_{1}}{\gamma_{1}}.
\end{align*}
If two random variables $X_{0},X_{1}$ obey $\mu_{0}$ and $\mu_{1}$,
respectively, we also use $\KL{X_{0}}{X_{1}}$ to denote the KL-divergence
between the two distributions. The KL-divergence satisfies the following
chain rule:
\[
\KL{X_{0}Y_{0}}{X_{1}Y_{1}}=\KL{X_{0}}{X_{1}}+\EE_{x\sim X_{0}}\left[\KLfrac{Y_{0}\mid X_{0}=x}{Y_{1}\mid X_{1}=x}\right].\footnotemark\textsuperscript{,}\footnotemark
\]
\addtocounter{footnote}{-1} 
\footnotetext{
We refer the interested readers to~\cite{cover2006elements} for a complete treatment on these facts.
}
\addtocounter{footnote}{1} 
\footnotetext{Here, we use the fraction-like notation to also denote the
KL-divergence for aesthetics, as we are comparing two conditional distributions.
The numerator in the fraction-like notation corresponds to the first argument
in the standard notation.}

The following simple analytical fact will be useful for us. 
\begin{fact}
\label{fact:inequality-a} Given $x,p\in\RR,$ then 
\end{fact}

\begin{enumerate}
\item $(1+x)^{p}\ge1+xp,$ for any $x>-1,$ and $p\ge1$. 
\item $(1+x)^{p}\le1+xp,$ for any $x>-1,$ and $0\le p\le1$. 
\end{enumerate}
\begin{proof}
Let $g=(1+x)^{p}-1-xp$ be a function on $x.$ Then 
\begin{equation}
g'=p(1+x)^{p-1}-p.\label{eq:partial-g-partial-a}
\end{equation}
When $p>1$, (\ref{eq:partial-g-partial-a}) is negative for $x\in(-1,0)$
and nonnegative for $x\ge0$. Thus, $g$ is decreasing in the interval
$x\in(-1,0)$ and increasing in $(0,\infty)$. Plug $x=0$ into $g$,
we get $0$. Therefore, $(1+x)^{p}\ge1+xp$. When $0<p<1,$ (\ref{eq:partial-g-partial-a})
is positive for $x\in(-1,0)$ and nonpositive for $x\ge0.$ Hence
$g$ attains its maxima within $(-1,\infty)$ at point $x=0$. 
\end{proof}

\subsection*{Discrete Fourier analysis}

Let $f:\DC\to\{0,1\}$ be any Boolean function. We would often treat
$f$ as a function $f:\CC\to[0,1]$ or $f:\RR^{n}\to\RR$ by considering
its multilinear extension, i.e., 
\[
f(x)=\sum_{S\subseteq[n]}\hat{f}(S)\chi_{S},
\]
here $\chi_{S}$ is the abbreviation of $\prod_{j\in S}x_{j}.$ An
important observation is that under this notation, 
\[
f(0)=\EE_{x\in\{-1,1\}^{n}}[f(x)].
\]
The set $\{\chi_{S}\}_{S\subseteq[n]}$ is a complete orthogonal basis
of the space $\RR^{\{-1,1\}^{n}}.$ Further, 
\[
2^{-n}\langle\chi_{S},\chi_{T}\rangle=\begin{cases}
1 & S=T,\\
0 & S\not=T.
\end{cases}
\]
Thus, for any $f,g:\{-1,1\}^{n}\to\RR,$ we have the following by
Parseval's identity and Plancherel Theorem, 
\begin{align}
 & 2^{-n}\langle f,g\rangle=\sum_{S\subseteq[n]}\hat{f}(S)\hat{g}(S).\label{eq:parseval}\\
 & \EE_{x\in\{-1,1\}^{n}}[f^{2}]=\sum_{S\subseteq[n]}\hat{f}(S)^{2}.\label{eq:plancherel}
\end{align}
Adopt the following notations for partial derivatives and the vector
differential operator, 
\begin{align*}
 & \partial_{i}f(x)=\sum_{S\ni i}\hat{f}(S)\chi_{S\setminus\{i\}},\\
 & \nabla f=(\partial_{1}f,\partial_{2}f,\ldots,\partial_{n}f).
\end{align*}
For functions on Boolean cubes, by considering their multilinear extensions
it's easy to see that the above definitions work exactly as expected:
For any $\delta\in\RR,$ 
\[
f(x+\delta e_{i})-f(x)=\delta\cdot\partial_{i}f(x).
\]
An important fact about the weight of the Fourier coefficients is
the following inequality, often referred to as the Level-1 inequality.
\begin{thm}[Level-1 inequality~\cite{talagrand1996much}]
\label{thm:level-1-ineq} Let $f:\CC\to\{0,1\}$ be the multilinear
extension of a Boolean function. Then for some absolute constant $C$,
we have
\[
|\nabla f(0)|^{2}\le Cf(0)^{2}\log\frac{e}{f(0)}.
\]
\end{thm}

We adopt the following standard definitions of the individual influence
and the max influence of function $f$: 
\begin{align*}
 & \INF_{i}(f)=\EE_{x\in\DC}[\partial_{i}f(x){}^{2}].\\
 & \mINF(f)=\max_{i\in[n]}\INF_{i}(f).
\end{align*}
By Plancherel Theorem, 
\[
\INF_{i}(f)=\sum_{S\subseteq[n]:\,i\in S}f(S)^{2}.
\]
The variance of $f$ is the following 
\[
\Var[f]=\EE_{x\in\DC}[f^{2}]-\EE_{x\in\DC}[f]^{2}.
\]
It is clear that 
\[
\Var[f]\le\sum_{i}^{n}\INF_{i}(f).
\]
Below is a straightforward corollary of the above inequality. 
\begin{fact}
\label{fact:mINF-bound}If $\Var[f]=2^{-o(n)},$ then
\[
\mINF(f)=2^{-o(n)}.
\]
\end{fact}

\subsection*{Martingales}

Recall that a discrete-time martingale is a sequence of random variables
$X_{0},X_{1},X_{2},\ldots,$ that satisfies 
\begin{itemize}
\item For each $n=0,1,2,\ldots,$ $\EE[|X_{i}|]<\infty$. 
\item For any $m<n$, $\EE[X_{m}\mid X_{n}]=X_{n}.$ 
\end{itemize}
A continuous-time martingale is a stochastic process $(X_{t})_{t\ge0}$
such that 
\begin{itemize}
\item For any $t$, $\EE[|X_{t}|]<\infty$. 
\item For any $s<t$, $\EE[X_{t}\mid X_{s}]=X_{s}.$ 
\end{itemize}
A submartingale is a stochastic process with the second property from
the above definition replaced by 
\[
\EE[X_{t}\mid X_{s}]\ge X_{s}.
\]

\begin{fact}
\label{fact:martingales}Let $X_{t},Y_{t}$ be martingales. 
\begin{enumerate}
\item $aX_{t}+bY_{t}$ and $X_{t}\cdot Y_{t}$ are also martingales for
any constant $a,b$. Hence any multilinear function of martingales
is a martingale. 
\item If $f:\RR\to\RR$ is a convex function, then the process $f(X_{t})$
is a submartingale. 
\end{enumerate}
\end{fact}

The stopping time $\tau$ of a stochastic process is a random variable
such that the event $\{\tau\le t\}$ is completely determined by $X_{\le t}$.
Given two stopping times $\tau_{1},\tau_{2},$ let $\tau_{1}\land\tau_{2}$
denote the new stopping time $\min\{\tau_{1},\tau_{2}\}$. For martingales,
we have the optional stopping theorem. 
\begin{thm}[Stopping Theorem]
If $\tau$ is almost surely bounded, then 
\[
\EE[X_{\tau}]=\EE[X_{0}].
\]
\end{thm}

For submartingales, the equality is replaced by a greater-than inequality.
Finally, the following inequalities will be useful
for us.
\begin{thm}[Doob's martingale inequality]
\label{thm:Doob-ineq}Let $X$ be a submartingale taking real values.
Then for any constant $C\ge0,$ 
\[
\PP\left[\sup_{0\le t\le T}X_{t}\ge C\right]\le\frac{\EE[\max\{X_{T},0\}]}{C}.
\]
\end{thm}

\begin{thm}[Concentration inequality~{\cite[Theorem 2.21]{chung2006complex}}]
\label{thm:concentration}Let $X_{1},X_{2},\ldots,X_{n}$ be martingales
with filtration $\cF$, such that for $i=1,2,\ldots,n$
\begin{align*}
 & \Var[X_{i}\mid\cF_{i-1}]\le\sigma_{i}^{2},\\
 & X_{i}-X_{i-1}\le M.
\end{align*}
Then
\begin{align*}
\PP[X\ge\lambda] & \le\exp\left(-\frac{\lambda^{2}}{2\sum\sigma_{i}^{2}+2M\lambda/3}\right).
\end{align*}
\end{thm}

Finally, we should warn the readers that in this paper, often $X$
is a vector and the subscripts are used for coordinates. In that case,
the random process $X$ is denoted $X(t),$ and $X_{i}(t)$
denotes the evolution of each individual coordinate.

\section{Controlled Process\label{sec:control}}

Fix a function $f:\{-1,1\}^{n}\to\{0,1\}$, and we view $f:\RR^{n}\to\RR$
by considering its multilinear extension. We assume that $f$ is not
a constant function. Therefore $f(0)>0$. In this section, we will
consider three different discrete random processes. 
The first one is the uniform process $X(t)\in\{-1,0,1\}^n$ for $t\in\{0,1,\ldots,n\}$.
It's called the uniform process because $X(n)$ will be a uniformly random
string from $\DC$. The second process $Y(t)$ is obtained from $X(t)$ by 
conditioning on $f(X(n))=1$. Therefore, we call $Y(t)$ the conditioned
process. The third process is in effect the same as the second process. They have
identical distributions. However, we will take the control theory perspective, 
and give a player a small number of random coordinates to control. We show that the player will
be able to alter a process to the conditioned process, which is otherwise
the uniform process. Therefore we sometimes call the third process the controlled
process.

First, we consider the following uniform process $X(t)$
for $t=0,1,2,\ldots,n$, such that $X(t)\in\{-1,0,1\}^{n}$, and $X(0)=0^{n}$. 

\vspace{1.5mm}
\noindent\fbox{\begin{minipage}[t]{1\textwidth - 2\fboxsep - 2\fboxrule}%
\textbf{\uline{Procedure 1}} (To generate the discrete uniform
process $X(t)$): \vspace{0.5mm}

Sample a uniformly random permutation $\pi:[n]\to[n].$

For time $t=1,2,\ldots,n$ 
\begin{itemize}
\item Let $i=\pi(t).$ Set $X_{i}(t)$ to be $-1$ or $1$ uniformly at
random. 
\item For all $j\in[n]\setminus\{i\}$, set $X_{j}(t):=X_{j}(t-1)$. 
\end{itemize}
\end{minipage}}

\vspace{1.5mm}
Clearly, the above process is just another way to sample a random
element from $\{-1,1\}^{n}$. We use the notation $P$ to denote the
probability measure over the paths of the above process. The subscript
$P$ will be used to emphasize the underlying process and the corresponding
measure. For example, $\Exp_{P}[f],\Pr_{P}[\cE]$ are the expectation
of the function $f$ and the probability of the event $\cE$, respectively,
both defined over the space of the paths of the above process $X(t)$.
A crucial component of our analysis is that all the partial derivatives
of $f(X(t))$ will be small with high probability for $t$ even very
close to $n$. We formulate it as the following lemma, whose proof
requires some technical preparations, and is therefore deferred to Section~\ref{subsec:proof-influence-small}.
\begin{lem}
\label{lem:influence-remain-small-P}Let $\epsilon>0$\footnote{Throughout this section, let's assume that $\epsilon n$ is a positive
integer.} be such that 
\[
\frac{16}{\epsilon}\ln\frac{4}{\epsilon}\le\ln\frac{1}{\mINF(f)}.
\]
Then for any $\theta\in(0,1)$,
\begin{align*}
\Pr_{P}\left[\max_{0\le t\le(1-\epsilon)n}|\partial_{i}f(X(t))|\ge\theta\right] & \le\theta^{-3}\mINF(f)^{\frac{\epsilon}{16}}+\exp(-\epsilon n/8).
\end{align*}
\end{lem}

Next, we modify Procedure 1 to generate what we call the conditioned process.
The goal is to guarantee that
the new process ends up being a random element sampled from $f^{-1}(1)$.
We use $Y(t)$ to distinguish this new process from $X(t)$. Let $Q$
be a new probability measure defined by the equation
\begin{equation}
\Pr_{Q}[Y_{i}(t)=\pm1\mid Y(t-1),\pi(t)]:=\frac{1}{2}\pm\frac{\partial_{i}f(Y(t-1))}{2f(Y(t-1))}.\label{eq:defQ}
\end{equation}
A calculation shows that the Radon-Nykodym derivative of the two measures
satisfies that for any realization $y(1),y(2),\ldots,y(s)\in\{-1,0,1\}^{n}$
of the process $Y(t)$ up to time $s$,
\begin{align}
\frac{dQ\bigl((y(t))_{1\leq t\leq s}\bigr)}{dP\bigl((y(t))_{1\leq t\leq s}\bigr)} & =\prod_{t=1}^{s}2\Pr_{Q}[Y_{\pi(t)}(t)=y_{\pi(t)}(t)\mid Y(t-1)=y(t-1)]\nonumber \\
 & =\prod_{t=1}^{s}\left(1+y_{\pi(t)}(t)\frac{\partial_{\pi(t)}f(y(t-1))}{f(y(t-1))}\right)\nonumber \\
 & =\prod_{t=1}^{s}\frac{f(y(t))}{f(y(t-1))}\nonumber \\
 & =\frac{f(y(s))}{f(0)}.\label{eq:Radon-Nykodym-derivative}
\end{align}
By taking $s=n$ above, we see that
the process $Y(t)$ according to $Q$ is equivalent to the same process
$X(t)$ according to $P$, only conditioned on the event that $f(X(n))=1$.
In particular, according to $Q$, $Y(n)$ is just a uniformly random
element from $f^{-1}(1)$. Further, if we sample $Y(t)$ and let $\cR(t)$ be
the restriction induced by $Y(t)$, then $(f|_{\cR(t)})^{-1}(1)$
is nonempty for any $t$ as long as $f$ is not the constant $0$
function. We record this simple but useful observation that $Q$ is a mild change
of measure of $P$.
\begin{claim}
\label{claim:closeness-P-Q}Let $\cE_{t}$ be some event that depends
only on the paths of the random process up to time $t$, e.g., $X(1),X(2),\ldots,X(t)$
according to $P$ or $Y(1),Y(2),\ldots,Y(t)$ according to $Q$. Then
for any $t\in[n],$
\[
\Pr_{Q}[\cE_{t}]\le\frac{\Pr_{P}[\cE_{t}]}{f(0)}.
\]
\end{claim}

\begin{proof}
This is immediate from~(\ref{eq:Radon-Nykodym-derivative}), 
\[
\Pr_{Q}[\cE_{t}]=\Exp_{P}\left[\II\{\cE_{t}\}\cdot\frac{dQ}{dP}\right]\le\frac{\Pr_{P}[\cE_{t}]}{f(0)}.\qedhere
\]
\end{proof}
\noindent We summarize the distribution of the ``conditioned'' process $Y(t)$ according
to $Q$:

\vspace{1.5mm}
\noindent\fbox{\begin{minipage}[t]{1\columnwidth - 2\fboxsep - 2\fboxrule}%
\textbf{\uline{Procedure 2}} (To generate the conditioned process $Y(t)$): \vspace{0.5mm}

Sample a uniformly random permutation $\pi:[n]\to[n].$

For time $t=1,2,\ldots,n$: 
\begin{itemize}
\item Let $i=\pi(t).$ Set $Y_{i}(t)$ according to the following distribution
\begin{align*}
 & \PP[Y_{i}(t)=\pm1]=\frac{1}{2}\pm\frac{\partial_{i}f(Y(t-1))}{2f(Y(t-1))},
\end{align*}
\item For all $j\in[n]\setminus\{i\}$, set $Y_{j}(t):=Y_{j}(t-1)$. 
\end{itemize}
\end{minipage}}

\vspace{1.5mm}

\subsection*{A control-theory point of view}

\noindent The next step will be to consider the above process as a
\emph{controlled version} of the conditioned process $Y(t)$. Fix $\epsilon>0$ and
consider the control problem where at each time $t$, with probability
$1-\epsilon$, $Y(t)$ does a uniformly random step (according to Procedure
1), and with probability $\epsilon$ a player gets to determine the sign
of $Y_{\pi(t)}$ according to her own choosing. 

The key observation of this section is that as long as the player
can control a small fraction of random coordinates, she can simulate
the conditioned process exactly. The motivation to study this controlled 
version of $Y(t)$ is the following:
The randomly fixed coordinates out of the player's control induces a
random restriction of the function $f$. If the player can assign the values
to the coordinates of her control, that is the alive coordinates of the
corresponding random restriction, to end up in $f^{-1}(1)$,
this means the restricted function has a nonempty preimage of $1$.

To this end, we consider the following procedure
(see \hyperref[Procedure-Y]{Procedure $\Pi$}) that generates the conditioned
process $Y(t)$ as well as the uniform process $X(t)$. 

The Procedure $\Pi$ starts with a sampling subroutine as the preparation
stage, then followed by two phases that generate $Y(t)$ for time
$t$ from $0$ to $n$. The first phase corresponds to that described
in the first paragraph of this section. During this phase the player
needs to cherish her rare opportunity and play ``aggressively.''
The second phase starts at a point of time $\tau$ when the aggressive
strategy no longer works. However, we will show that $\tau$ is very
close $n$ with high probability. As a result, it would not be a problem
to give the player full control from now on and let her play ``safely'' till the end.

\vspace{1.5mm}

\noindent\fbox{\begin{minipage}[t]{1\textwidth - 2\fboxsep - 2\fboxrule}%
\textbf{\uline{Procedure \mbox{$\Pi$}\label{Procedure-Y}}} (The
controlled version of processes $Y(t)$ and $X(t)$): \vspace{1mm}

\emph{\# Sampling Subroutine}
\begin{itemize}
\item Sample a uniformly random permutation $\pi:[n]\to[n].$
\item Sample a set $T\subseteq\{1,2,\ldots,n\},$ such that independently for each $i\in[n],$
\[
\PP[i\in T]=\epsilon.
\]
$T$ will be the set of times when the player gets to determine the
value of the coordinate. 
\item Sample a uniformly random $z\in\mopo^{\pi\bar{T}},$ the random assignment
to the variables not controlled by the player.
\end{itemize}
\vspace{3mm}

\emph{\# Phase 1}

Set $Y(0)=X(0)=0^{n}$. 

For time $t=1,2,\ldots,n$:
\begin{itemize}
\item Let $i=\pi(t).$
\item (Coordinate picked uniformly) If $t\not\in T$, set $Y_{i}(t)=X_{i}(t)=z_{i}$. 
\item (Coordinate determined by player) If $t\in T$, set $Y_{i}(t)$ and
$X_{i}(t)$ according to the following distributions 
\begin{align*}
 & \PP[Y_{i}(t)=\pm1]=\frac{1}{2}\pm\frac{1}{2\epsilon}\cdot\frac{\partial_{i}f(Y(t-1))}{f(Y(t-1))},\\
 & \PP[X_{i}(t)=\pm1]=\frac{1}{2}.
\end{align*}
\item For all $j\in[n]\setminus\{i\}$, set $Y_{j}(t):=Y_{j}(t-1)$, $X_{j}(t):=X_{j}(t-1)$.
\item If either of the following holds, \textbf{exit} this loop 
\begin{align*}
 & \max_{i\in[n]}|\partial_{i}f(Y(t))|>\epsilon\delta,\quad f(Y(t))<\delta. & \text{\emph{\#\text{ the breaking condition}}}
\end{align*}
\end{itemize}
\vspace{3mm}

\emph{\# Phase 2}

While $t<n$:
\begin{itemize}
\item $t=t+1.$
\item Let $i=\pi(t).$ Set $Y_{i}(t)$ and $X_{i}(t)$ according to the
following distributions 
\begin{align*}
 & \PP[Y_{i}(t)=\pm1]=\frac{1}{2}\pm\frac{\partial_{i}f(Y(t-1))}{2f(Y(t-1))},\\
 & \PP[X_{i}(t)=\pm1]=\frac{1}{2}.
\end{align*}
\item For all $j\in[n]\setminus\{i\}$, set $Y_{j}(t):=Y_{j}(t-1)$, $X_{j}(t):=X_{j}(t-1)$.
\end{itemize}
\vspace{3mm}

\textbf{Output} $\{Y(t)\}_{t\in\{0,1,\ldots,n\}},\{X(t)\}_{t\in\{0,1,\ldots,n\}}.$%
\end{minipage}}

\vspace{1.5mm}

The process $Y(t)$ will be the main process with which our analysis
concerns, whereas the process $X(t)$ is only defined for the sake
of entropy comparison: We will later argue that the KL-divergence
between the two processes is not too large. It is evident that in
Phase 1 the distribution of $Y(1),Y(2),...$ according to Procedure
$\Pi$ is identical to its distribution according to measure $Q$
as long as 
\begin{equation}
|\partial_{\pi(t)}f(Y(t-1))|\le\epsilon f(Y(t-1)).\label{eq:stopping-condition}
\end{equation}
Indeed, if~(\ref{eq:stopping-condition}) holds, we have 
\begin{align*}
\PP[Y_{\pi(t)}(t)=\pm1] & =(1-\epsilon)\frac{1}{2}+\epsilon\left(\frac{1}{2}\pm\frac{1}{2\epsilon}\cdot\frac{\partial_{\pi(t)}f(Y(t-1))}{f(Y(t-1))}\right)\\
 & =\frac{1}{2}\pm\frac{\partial_{\pi(t)}f(Y(t-1))}{2f(Y(t-1))}.
\end{align*}
Let time $\tau$ be $n+1$ if the breaking condition is never hit,
otherwise let $\tau$ be the time when the breaking condition is hit.
The reader may wonder that a more natural choice of the ``breaking''
condition would be the violation of~(\ref{eq:stopping-condition}).
Our definition forces that (i) $f(Y(\tau-1))$ is large, in addition
to that (ii) all derivatives $|\partial_{i}f(Y(\tau-1))|$ is small
compared to the magnitude of $f(Y(\tau-1))$. Both facts will be very
useful in later sections. Formally, we summarize our definition of
$\tau$ as below, 
\begin{equation}
\tau=\tau_{1}\land\tau_{2}\land(n+1),\label{eq:stopping-time}
\end{equation}
where
\begin{align*}
 & \tau_{1}=\min\{t:\max_{i\in[n]}|\partial_{i}f(Y(t))|>\epsilon\delta\}.\\
 & \tau_{2}=\min\{t:f(Y(t))<\delta\}.
\end{align*}
The values of the parameters $\epsilon,\delta$ will be specified
later on. By definition, the condition~(\ref{eq:stopping-condition})
holds for $t<\tau$. We should think $\tau$ as a stopping time of
Phase 1. After the stopping time $\tau$, the player gets to control
each coordinates left. She simply assigns the values according $Q$
as in Procedure 2. Since in both phases Procedure $\Pi$ has the same
law as that of $Q$, the controlled process $Y(t)$ defined in procedure $\Pi$
is identical in distribution to the conditioned process $Y(t)$ defined in Procedure
$2$. The same is clearly true for the uniform process $X(t)$ in its two
versions (Procedure $1$ and Procedure $\Pi$).

In the preparation stage, Procedure $\Pi$ samples a random
permutation $\pi,$ a set $T$ of times controlled by the player and
$z$ the random assignment to the variables not controlled by the
player. For every $m\in[n]$, let $\cG_{m}$ be the $\sigma$-algebra generated by
\[
\pi|_{\{1,2,\ldots,m\}},\quad T\cap\{1,2,\ldots,m\},\quad\text{and }z|_{\pi\{1,2,\ldots,m\}}.
\]
Thus, $\cG_{m}$ contains all the information in a run of Procedure
$\Pi$, excluding the player's choices, up to time $m$. Also, $\cG_m$ induces
a restriction of $T$:\footnote{We can also consider any restriction
$\cR=(S,z)$ for $S\subseteq \pi(\{1,2,\ldots,m\}\setminus T)$. We will use this observation in later sections.}
\[\cR=(\pi(\{1,2,\ldots,m\}\setminus T)),z).\]
A moment's thought reveals that if the controlled process $Y(t)$ in a run of Procedure $\Pi$ satisfies
that $\tau>m$, then $f|_{\cR}$ contains a nonempty preimage of 1.
\begin{claim}
If $\PP[ \tau > m \mid \cG_m] >0$, then
$(f|_{\cR})^{-1}(1) \not=\varnothing.$
\end{claim}
Therefore, if we can argue that $\tau>m$ running Procedure $\Pi$ on $f$ and $1-f$ with the same $\cG_m$,
then we actually proved that $f|_{\cR}$ is nonconstant. To give a quantitative bound on the variance
of $f|_\cR$ requires some more work. The above discussion sets two tasks for the remainder of
this section. First, to analyze the stopping time $\tau$ and second, to provide the necessary tools
to bound the variance of the restricted function.

\subsection{Stopping time \texorpdfstring{$\tau$ of the process $Y(t)$}{of the controlled process}  }

Next, we prove that with high probability $\tau>(1-\epsilon)n$ for
very small $\epsilon$. Therefore, Phase 2 in Procedure $\Pi$ can
not be too long. 
\begin{lem}[Stopping time $\tau$ of the process $Y(t)$]
\label{lem:stopping-time-tau} Let $f:\DC\to\{0,1\}$ be such that
$\Var[f]\ge2^{-o(n)}.$ Further, let $\epsilon>0$ and $\delta$ be
such that
\begin{align*}
 & \frac{16}{\epsilon}\ln\frac{4}{\epsilon}\le\ln\frac{1}{\mINF(f)},\\
 & \delta\ge\frac{\mINF(f)^{\epsilon/80}}{\epsilon}.
\end{align*}
Then for sufficiently large $n$, we have 
\[
\Pr_{Q}[\tau\le(1-\epsilon)n]\le\frac{3\delta}{f(0)}.
\]
\end{lem}

\begin{proof}
The proof relies on the fact that $Q$ is a mild change of measure
with respect to $P$. Consider the following two bad events, 
\begin{align*}
\cE_{1}:\quad\tau_{1}\le(1-\epsilon)n,\\
\cE_{2}:\quad\tau_{2}\le(1-\epsilon)n.
\end{align*}
We first bound $\PP_{Q}[\cE_{1}]$. Note that
\begin{align*}
\Pr_{P}[\cE_{1}] & =\Pr_{P}\left[\max_{0\le s\le(1-\epsilon)n}|\partial_{i}f(X(s))|\ge\epsilon\delta\right]\\
 & \le\Pr_{P}\left[\max_{0\le s\le(1-\epsilon)n}|\partial_{i}f(X(s))|\ge\mINF(f)^{\epsilon/60}\right]\\
 & \le\mINF(f)^{\epsilon/80}+\exp(-\epsilon n/8),
\end{align*}
where the second step holds as $\epsilon\delta\ge\mINF(f)^{\epsilon/80}\ge\mINF(f)^{\epsilon/60}$;
the final step applies Lemma~\ref{lem:influence-remain-small-P}.
The above bound in turn by Claim~\ref{claim:closeness-P-Q} implies
that
\begin{align}
\Pr_{Q}[\cE_{1}] & \le\frac{\Pr_{P}[\cE_{1}]}{f(0)}\nonumber \\
 & \le\frac{\mINF(f)^{\epsilon/80}+\exp(-\epsilon n/8)}{f(0)}.\label{eq:controlled-bound-E1}
\end{align}
Next we move to bound $\PP_{Q}[\cE_{2}]$. It is immediate from Claim~\ref{claim:closeness-P-Q}:
\begin{equation}
\Pr_{Q}[\cE_{2}]\le\frac{\delta}{f(0)}.\label{eq:controlled-bound-E2}
\end{equation}
Apply union bound to~(\ref{eq:controlled-bound-E1}) and~(\ref{eq:controlled-bound-E2}),
then for large enough $n$,
\begin{align*}
\Pr_{Q}[\tau\le(1-\epsilon)n] & =\Pr_{Q}[\cE_{1}\lor\cE_{2}]\\
 & \le\frac{1}{f(0)}\cdot\left(\delta+\mINF(f)^{\epsilon/80}+\exp(-\epsilon n/8)\right)\\
 & \le\frac{3\delta}{f(0)}
\end{align*}
where in the final step we have $\delta\ge\mINF(f)^{\epsilon/80}=\exp(-o(\epsilon n))$
by Fact~\ref{fact:mINF-bound}. 
\end{proof}

\subsection{\label{subsec:KL-control-process}The KL-divergence between \texorpdfstring{$Y(t)$
and $X(t)$}{the controlled and the discrete uniform process}}

The purpose
of this subsection is to show that for any $m\in[n]$, $Y(n)$ given $\cG_{m}$
is close to uniform with high probability over the random choices
associated with $\cG_{m}$. In particular, we will show that the KL-divergence
between $Y(n)$ and $X(n)$ given $\cG_{m}$ is small with high probability
over the random choices associated with $\cG_{m}.$ Recall that the
coordinates of $X(n)$ not fixed by $\cG_{m}$ are uniform.
\begin{lem}
\label{lem:Q-KL-bound}For any $m\in [n]$, abbreviate
\begin{align*}
 & \cG_{m}=(\pi|_{\{1,2,\ldots,m\}},T\cap\{1,2,\ldots,m\},z|_{\pi\{1,2,\ldots,m\}}).
\end{align*}
Then for some universal constant $C$, and $\epsilon$, $\delta$ in the
breaking condition in~\hyperref[Procedure-Y]{Procedure $\Pi$},
\begin{align*}
 & \PP_{\cG_{m}}\left[\KLfrac{Y(n)\mid\cG_{m}}{X(n)\mid\cG_{m}}\ge\frac{C}{\epsilon}\ln\frac{e n}{n-m+1}\log\frac{e}{\delta}\right]\le\delta.
\end{align*}
\end{lem}

\begin{proof}
Let $\tau'=\tau\land(m+1)$. By definition of the stopping time $\tau$~(\ref{eq:stopping-time}),
for $t<\tau',$
\begin{align}
 & f(Y(t))\ge\delta,\label{eq:KL-bound-mean}\\
 & |\partial_{i}f(Y(t))|\le\epsilon f(Y(t)).\label{eq:KL-bound-derivative}
\end{align}
We calculate the KL-divergence between $Y(n)\mid\cG_{m}$ and $X(n)\mid\cG_{m}$.
By the chain rule,
\begin{align*}
\KLfrac{Y(n)\mid\cG_{m}}{X(n)\mid\cG_{m}} & =\EE_{\cG_{m}}\left[\sum_{t=1}^{\tau'-1}\II\{t\in T\}\KLfrac{Y_{\pi(t)}(t)\mid Y(t-1)}{X_{\pi(t)}}\right.\\
 & \qquad\qquad\left.+\KLfrac{Y(n)|_{\pi\{\tau',\tau'+1,\ldots,n\}}\mid Y(\tau'-1)}{X(n)|_{\pi\{\tau',\tau'+1,\ldots,n\}}}\right],
\end{align*}
where the equality holds because for any $t\in T$, 
\[
(Y_{\pi(t)}\mid Y(t-1))=(Y_{\pi(t)}\mid Y(t-1),\cG_{m}),
\]
namely, any variable $Y_{\pi(t)}(t)$ controlled by the player is
independent of the variables in the future that she has no control
of; and all coordinates in $X(n)$ are independent. 

Next, using formula~\eqref{eq:Radon-Nykodym-derivative}, combined
with~(\ref{eq:KL-bound-mean}), it follows that for any $t\le\tau'$,
\begin{align}
 & \KLfrac{\left.\left(Y(n)|_{\pi\{t,t+1,\ldots,n\}}\right)\;\right|Y(t-1)}{X(n)|_{\pi\{t,t+1,\ldots,n\}}}\nonumber \\
 & \qquad\qquad\qquad=\log\frac{dQ((Y(i))_{t\leq i\leq n})}{dP((X(i))_{t\leq i\leq n})}=\log\frac{1}{f(Y(t-1))}\le\log\frac{1}{\delta}.\label{eq:KL-after-stopping-time}
\end{align}
Combining the above two displays,
\begin{align*}
 & \KLfrac{Y(n)\mid\cG_{m}}{X(n)\mid\cG_{m}}-\log\frac{1}{\delta}\\
 & \qquad\le\sum_{t=1}^{n}\EE_{Y(t-1)\mid\cG_{m}}\left[\II\{t\in T\}\II\{t<\tau'\}\KLfrac{Y_{\pi(t)}(t)\mid Y(t-1)}{X_{\pi(t)}(t)}\right]\\
 & \qquad=\sum_{t=1}^{n}\EE\left[\II\{t\in T\}\II\{t<\tau'\}\left(1-H\!\left(\frac{1}{2}+\frac{\partial_{\pi(t)}f(Y(t-1))}{2\epsilon f(Y(t-1))}\right)\right)\right]\\
 & \qquad\le\sum_{t=1}^{n}\EE\left[\II\{t\in T\}\II\{t<\tau'\}\frac{1}{\epsilon^{2}}\left(\frac{\partial_{\pi(t)}f(Y(t-1))}{f(Y(t-1))}\right)^{2}\right],
\end{align*}
where the second step is by the definition of the KL-divergence; and
the final step is due to~\eqref{eq:entropy-upper-bound}. Abbreviate
\[
Z_{t}:=\II\{t\in T\}\II\{t<\tau'\}\left(\frac{\partial_{\pi(t)}f(Y(t-1))}{f(Y(t-1))}\right)^{2}.
\]
\begin{claim}
There is some universal constant $C\ge1,$ such that for any $t<\tau',$
\begin{align}
 & Z_{t}\mid Y(t-1)\in[0,\epsilon^{2}],\label{eq:Z-bounded}\\
 & \EE[Z_{t}\mid Y(t-1)]\le\frac{C\epsilon}{n-t+1}\log\frac{e}{\delta},\label{eq:Z-mean}\\
 & \Var[Z_{t}\mid Y(t-1)]\le\epsilon^{2}\EE[Z_{t}\mid Y(t-1)].\label{eq:Z-variance}
\end{align}
\end{claim}

\begin{proof}
(\ref{eq:Z-bounded}) follows from~(\ref{eq:KL-bound-derivative}).
Let 
\[
v(t):=\frac{(\nabla f(Y(t-1)))|_{\pi\{t,t+1,\ldots,n\}}}{f(Y(t-1))}.
\]
Then, 
\begin{align*}
\EE[Z_{t}\mid Y(t-1)] & =\EE_{\pi(t)}[\epsilon v(t)_{\pi(t)}^{2}\mid Y(t-1)]\\
 & =\frac{\epsilon|v(t)|^{2}}{n-t+1}\\
 & \le\frac{C\epsilon}{n-t+1}\log\frac{e}{f(Y(t-1))}\\
 & \le\frac{C\epsilon}{n-t+1}\log\frac{e}{\delta},
\end{align*}
where the first step holds as $t\in T$ with probability $\epsilon$;
in the second step, $\pi(t)$ is random within the $n-t+1$ alive
coordinates given $Y(t-1)$; the third step follows the Level-1 inequality
of Theorem~\ref{thm:level-1-ineq}; and the final step is due to~(\ref{eq:KL-bound-mean}). 

The variance of $Z_{t}\mid Y(t-1)$ can be bounded as follows:
\begin{align*}
(\epsilon^{2}-\EE[Z_{t}])\EE[Z_{t}]-\Var[Z_{t}] & =\EE[(\epsilon^{2}-Z_{t})Z_{t}]\ge0.
\end{align*}
We comment that such a bound is sometimes referred to as the Bhatia-Davis
inequality. 
\end{proof}
By~(\ref{eq:Z-mean})-(\ref{eq:Z-variance}), the definition that $\tau'\le m+1,$ and the
following elementary fact that
\[
\ln(n+1)\le\sum_{i=1}^{n}\frac{1}{n}\le\ln en,
\]
we have
\begin{align}
 & \sum_{t=1}^{n}\EE[Z_{t}\mid Y(t-1)]\le \lambda,\label{eq:sum-Z-mean}\\
 & \sum_{t=1}^{n}\Var[Z_{t}\mid Y(t-1)] \le \epsilon^2\lambda, \label{eq:sum-Z-var}
\end{align}
where
\[
\lambda=C\epsilon\ln\frac{en}{n-m+1}\log\frac{e}{\delta}.
\]
The lemma is concluded by estimating,
\begin{align*}
 & \PP\left[\KLfrac{Y(n)\mid\cG_{m}}{X(n)\mid\cG_{m}}\ge\frac{3\lambda}{\epsilon^{2}}+\log\frac{1}{\delta}\right]\\
 & \qquad\qquad\le\PP\left[\sum_{t=1}^{n}Z_{t}\mid Y(t-1)\ge3\lambda\right]\\
 & \qquad\qquad\le\exp\left(-\frac{(2\lambda)^{2}}{2\epsilon^{2}\lambda+4\epsilon^{2}\lambda/3}\right)\\
 & \qquad\qquad\le\exp\left(-\frac{C}{\epsilon}\ln\frac{en}{n-m+1}\log\frac{e}{\delta}\right)\le\delta,
\end{align*}
where the second step invokes the concentration inequality of Theorem~\ref{thm:concentration}
since $Z_{t}-\EE[Z_{t}\mid Y(t-1)]$ is a martingale with respect
to $Y(0),Y(1),\ldots,Y(t-1)$. This finishes our proof to Lemma~\ref{lem:Q-KL-bound}
with a change of the constant $C$.
\end{proof}

\section{Proofs of the Main Results\label{sec:appl}}

In this section, we prove a sharp ``it ain't over till it's over'' theorem, i.e., 
the  nonasymptotic version of Theorem~\ref{thm:mainintro}. Then we comment on its optimality,
and discuss its applications to block sensitivity and decision tree complexity. 
\begin{thm}[``It ain't over till it's over'']
\label{thm:main}There are absolute constant $C>1.$ Given $f:\{-1,1\}^{n}\to\{0,1\}$,
such that $\mINF(f)<1/C$ and $\Var[f]=2^{-o(n)}$. Let $\cR$ be
a random restriction that keeps exactly $\lceil\rho n\rceil$ variables
alive, where 
\[
\frac{C}{\Var[f]}\cdot\frac{\ln\ln(1/\mINF(f))}{\ln(1/\mINF(f))}\le\rho\le1.
\]
Let $p$ be such that
\begin{equation}
\frac{8\mINF(f)^{\rho/C}}{\rho\Var[f]}\le p\le1.\label{eq:main-delta-bound}
\end{equation}
Then for large enough $n,$
\begin{align}
 & \PP\left[\Var[f|_{\cR}]\le\exp\left(-\frac{C}{\rho}\ln\frac{e}{\rho}\cdot\log\frac{8e}{p\Var[f]}\right)\right]\le p.\label{eq:aintover-variance-bound}
\end{align}
\end{thm}

Before proving the above theorem, we record the following simple fact.
\begin{prop}
\label{fact:KL-to-mean}Let $f:\mopo^{n}\to\{0,1\}$ be a Boolean
function. Let $\mu$ be the uniform distribution and $\gamma$ be
some arbitrary distribution over $\mopo^{n}.$ If
\begin{align*}
 & \gamma(f)\ge\delta,\quad\KL{\gamma}{\mu}\le K.
\end{align*}
Then
\[
\mu(f)\ge2^{-(K+H(\delta))/\delta}.
\]
In particular, if $\gamma(f)=1,$ then
\[
\mu(f)\ge2^{-K}.
\]
\end{prop}

\begin{proof}
Assume that $\gamma(f)=\delta,$ and $\KL{\gamma}{\mu}=K.$ This is
without loss of generality because $2^{-(K+H(\delta))/\delta}$ is
decreasing in $K$ and increasing in $\delta$ by elementary calculus.
Let $\gamma_{0},\gamma_{1}$ be the uniform distributions over $f^{-1}(0)$
and $f^{-1}(1)$, respectively. Note $\delta\gamma_{1}+(1-\delta)\gamma_{0}=\EE_{\pi}[\gamma\circ\pi],$
where $\pi$ is taken over the product of permutations on $f^{-1}(0)$
and $f^{-1}(1).$ Thus by convexity,
\begin{align*}
\KL{\delta\gamma_{1}+(1-\delta)\gamma_{0}}{\mu} & \le\KL{\gamma}{\mu}.
\end{align*}
Consequently, let $\eta=\mu(f)$, then
\begin{align*}
 & \delta\log\frac{\delta}{\eta}+(1-\delta)\log\frac{1-\delta}{1-\eta}\le K\\
\implies\quad & \delta\log\frac{1}{\eta}+(1-\delta)\log\frac{1}{1-\eta}\le K+H(\delta)\\
\implies\quad & \delta\log\frac{1}{\eta}\le K+H(\delta)\\
\implies\quad & \eta\ge2^{-(K+H(\delta))/\delta}.\qedhere
\end{align*}
\end{proof}
Next, we set forth to prove Theorem~\ref{thm:main}. Set
\begin{align}
 & \epsilon=\max\left\{ \eta:\eta\le\frac{\rho}{3},\eta n\text{ is an integer}\right\} ,\label{eq:main-eps}\\
 & \delta=p\Var[f]/8.\label{eq:main-delta}
\end{align}
It's straightforward to verify that for large enough $n$, and large
enough constant $C$, we have
\begin{align}
 & \frac{16}{\epsilon}\ln\frac{4}{\epsilon}\le\ln\frac{1}{\mINF(f)},\label{eq:main-eps-mINF}\\
 & \delta \ge\frac{\mINF(f)^{\epsilon/80}}{\epsilon}.\label{eq:main-delta-inf}
\end{align}

We will run Procedure $\Pi$ described in Section~\ref{subsec:KL-control-process}
with the above setting of parameters $\epsilon$ and $\delta$. Recall that Procedure $\Pi$
first samples the random permutation $\pi$, the set $T$ of times
controlled by the player and $z\in\mopo^{\pi\bar{T}}$ is the random
assignment for $t\not\in T$ in Phase 1. Let $m=(1-\epsilon)n$, $U=\{1,2,\ldots,m\}\setminus T.$
Note that by a Chernoff bound, the probability that $|U|$ is less
than $\lfloor(1-\rho)n\rfloor$ is at most $\exp(-\epsilon n/2).$
Conditioning on that $|U|\ge\lfloor(1-\rho)n\rfloor,$ we randomly
sample a set $S$ of $\lfloor(1-\rho)n\rfloor$ elements from $U$. 

\begin{figure}[h]
\includegraphics[scale=0.35]{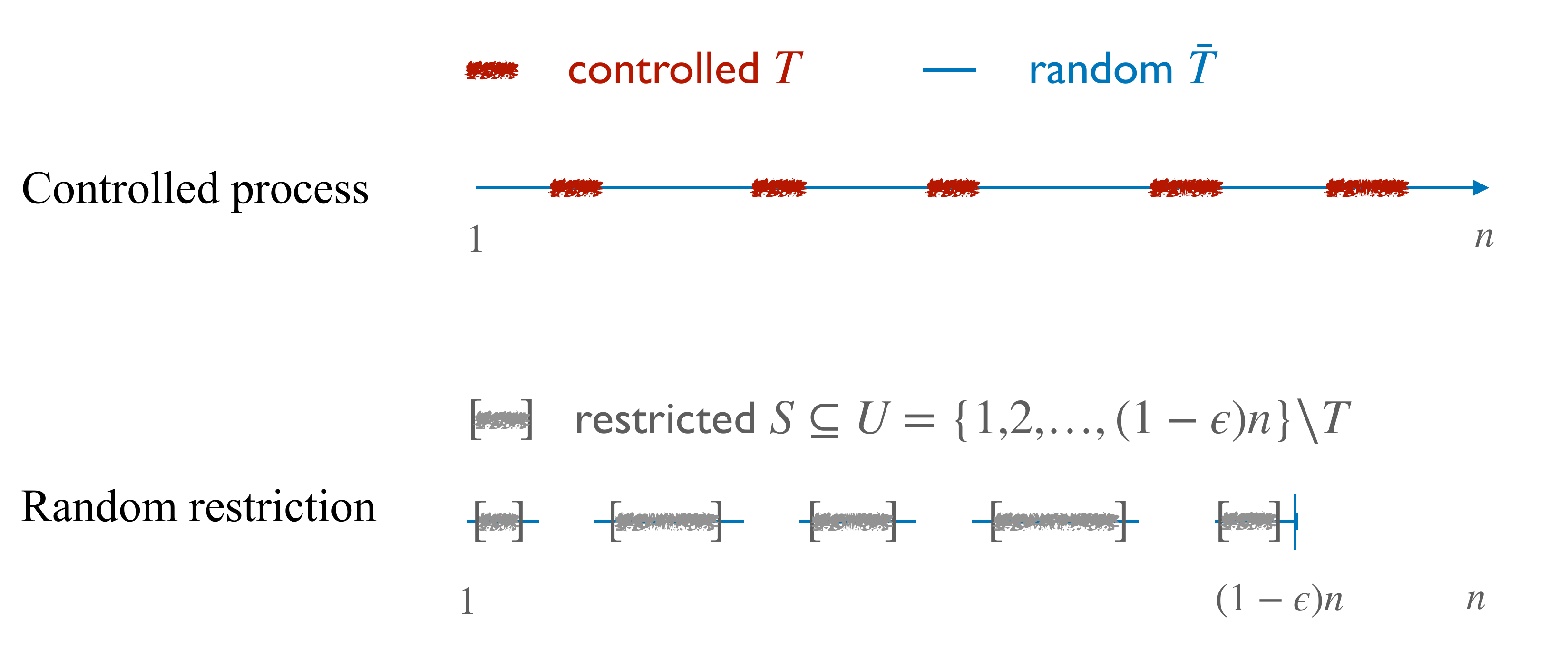}\caption{Random process and random restrictions}
\end{figure}
Consider the following event
\begin{align*}
 & \cE:=\{\tau>m\}\cap\{|U|\ge\lfloor(1-\rho)n\rfloor\}.
\end{align*}
The first event in this intersection can be bounded by Lemma~\ref{lem:stopping-time-tau}. Hence,
by union bound, 
\begin{align}
\PP[\neg\cE] & \le\frac{3\delta}{f(0)}+\exp(-\epsilon n/2), \nonumber\\
 & \le\frac{4\delta}{f(0)}, \label{eq:bad-event-main}
\end{align}
where 
the second step applies Fact~\ref{fact:mINF-bound} with~(\ref{eq:main-delta-inf}).
Conditioning on $\cE,$ $\cR=(S,Y(n))$ is distributed as a random
restriction that keeps $\lceil\rho n\rceil$ variables alive. Furthermore,
the restricted function $f|_{\cR}$ satisfies that its mean is bounded
away from 0 with high probability.
\begin{claim}
\label{claim:restriction-mean-bound}For some universal constant $C'>1$,
\begin{align}
 & \PP[\mu(f|_{\cR})<2^{-K}\mid\cE]\le\frac{2\delta}{\PP[\cE]},\label{eq:f-mean}
\end{align}
where 
\begin{equation}
K=\frac{C'}{\epsilon}\ln\frac{e}{\epsilon}\log\frac{e}{\delta}.\label{eq:main-K}
\end{equation}
\end{claim}

Our theorem follows immediately from the above claim. Indeed, if $\cE'$ is
defined analogously to $\cE$ where $f$ is replaced by $1-f$, we
have
\begin{align*}
 & \PP[\Var[f|_{\cR}]<2^{-K-1}]\\
 & \qquad\le\PP[\neg\cE\lor\neg\cE']+\PP[\mu(f|_{\cR})<2^{-K}\mid\cE]+\PP[\mu(f|_{\cR})>1-2^{-K}\mid\cE']\\
 & \qquad\le\frac{4\delta}{f(0)}+\frac{4\delta}{1-f(0)}+\frac{2\delta}{\PP[\cE]}+\frac{2\delta}{\PP[\cE']}\\
 & \qquad\le\frac{8\delta}{\Var[f]},
\end{align*}
where in the first step, note that $2^{-K}<1/2$; the second step plugs in~(\ref{eq:bad-event-main})-(\ref{eq:f-mean});
in the final step, note that by~(\ref{eq:main-delta}),
\[
\PP[\cE],\PP[\cE']\ge1-\frac{4\delta}{\Var[f]}=1-\frac{p}{2}\ge\frac{1}{2}>\Var[f].
\]
 In view of~(\ref{eq:aintover-variance-bound}) and~(\ref{eq:main-K}),
the proof to Theorem~\ref{thm:main} is finished. It remains to prove
Claim~\ref{claim:restriction-mean-bound}.
\begin{proof}[Proof of Claim~\emph{\ref{claim:restriction-mean-bound}}]
 Recall that $m=(1-\epsilon)n$ and $U=\{1,2,\ldots,m\}\setminus T.$
Abbreviate 
\[
\cG_{m}=(\pi|_{\{1,2,\ldots,m\}},T\cap\{1,2,\ldots,m\},z|_{\pi\{1,2,\ldots,m\}}),
\]
the information of the random process generated by Procedure $\Pi$
excluding the player's choices up to time $m$. Further, let
$\gamma$ be the distribution of $Y(n)|_{\pi\bar{U}}$ given $\cG_{m}$.
Then by definition of Procedure $\text{\ensuremath{\Pi}}$ running
with with respect to function $f$, $\gamma(f|_{(\pi U,z)})=1$. Hence,
\begin{align*}
 & \PP\left[\mu(f|_{(\pi U,z)})<2^{-K}\mid\cE\right]\\
 & \qquad\le\PP\left[\KLfrac{Y(n)\mid\cG_{m}}{X(n)\mid\cG_{m}}>K\;\middle|\;\cE\right]\\
 & \qquad\le\frac{\delta}{\PP[\cE]},
\end{align*}
where the first step is due to Proposition~\ref{fact:KL-to-mean}; the
second step follows Lemma~\ref{lem:Q-KL-bound} for a suitable constant
$C'$ in~(\ref{eq:main-K}). Now for any 
\[
S\in\binom{\{1,2,\ldots,m\}}{\lfloor(1-\rho)n\rfloor},\quad\text{and }y\in\mopo^{\pi S},
\]
let $\zeta(S,y)$ be the distribution of $(U,z)\mid\{(S\subseteq U)\land(z|_{\pi S}=y)\}.$
Then
\begin{align*}
 & \mu(f|_{(\pi S,y)})=\EE_{(U,z)\sim\zeta(S,y)}[\mu(f|_{(\pi U,z)})].
\end{align*}
Thus, by Markov's inequality $\mu(f|_{(\pi S,y)})<2^{-K-1}$ implies
that 
\[
\PP_{(U,z)\sim\zeta(S,y)}\left[\mu(f|_{(\pi U,z)})<2^{-K}\right]>\frac{1}{2}.
\]
Consequently,
\begin{align*}
 & \frac{1}{2}\PP_{\pi,S,y}[\mu(f|_{(\pi S,y)})<2^{-K-1}]\\
 & \quad\le\PP_{\pi,S,y,(U,z)\sim\zeta(S,y)}\left[\mu(f|_{(\pi U,z)})<2^{-K}\right]\\
 & \quad=\PP_{\pi,T,z}\left[\left.\mu(f|_{(\pi U,z)})<2^{-K}\;\right|\;\{|U|\ge\lfloor(1-\rho)n\rfloor\}\right]\\
 & \quad\le\frac{\delta}{\PP[\cE]}.
\end{align*}
In view of~(\ref{eq:f-mean}), we are done. 
\end{proof}

\begin{rem}
\label{rem:main}If we consider the random restriction that keeps
each variable alive independently with probability $\rho$, the same
statement holds with a slight
modification on the proof to the corresponding version of Claim~\ref{claim:restriction-mean-bound}.
\end{rem}

\subsection*{Optimality of our result}

Our Theorem~\ref{thm:main} is essentially optimal with respect to
$p$ and $\rho$. Consider the $(1-\rho)$-random restriction $\cR_{1-\rho}$.
First, we check the optimality in the regime when $\rho=\Omega((\log(1/\mINF(f)))^{-1}).$
Consider the majority function $\MAJ_{n}:\{-1,1\}^{n}\to\{0,1\},$
\[
\MAJ_{n}(x)=\begin{cases}
0 & \sum_{i\in n}x_{i}>0,\\
1 & \text{otherwise.}
\end{cases}
\]
It's well-known that $\mINF(\MAJ_{n})=\Theta(1/\sqrt{n}).$ For $\rho=\Omega(1/\log n),$
let $\cR_{1-\rho}=(S,X)$ be the random restriction. Say $|S|=n-k.$
With probability at least $1-\exp(-\Theta(\rho n))$, $k\in(0.5\rho n,2\rho n).$
Then by the Berry-Esseen Theorem, for $\lambda = O(\sqrt{(n-k)\log(n-k)})$,
\begin{align*}
 & \PP\left[\left|\sum_{i\in S}X_{i}\right| \ge\lambda \right]=\exp\left(-\Theta\left(\frac{\lambda^{2}}{n-k}\right)\right),\\
 & \Var\left[\MAJ_{n}|_{\cR_{1-\rho}} \;\middle\vert\;  \left\{ \left|\sum_{i\in S}X_{i}\right|\ge\lambda\right\}\right]\le\exp\left(-\Theta\left(\frac{\lambda^{2}}{k}\right)\right).
\end{align*}
Thus for $p=\Omega(1/\sqrt{n-k})$,
\[
\PP\left[\Var[\MAJ_{n}|_{\cR_{1-\rho}}]\le p^{\Theta\left(\frac{1}{\rho}\right)}\right]= p.
\]
Our bound on the variance is tight up to a $\log(1/\rho)$
factor in the exponent with respect to $\rho$.

Second, we check the optimality in the regime when $\rho=O((\log(1/\mINF(f)))^{-1}).$
Consider the tribes function $\TR_{n}:\{-1,1\}^{n}\to\{0,1\},$ 
\[
\TR{}_{n}:x\mapsto\AND_{n/w}\left(\cdots,\bigvee_{j=1}^{w}x_{ij},\cdots\right),
\]
where for any positive integer $w$, $n$ is the smallest integral
multiple of $w$ such that $\Pr[\TR_{n}(x)=1]\le1/2$; $\AND_{n/w}:\{0,1\}^{n/w}\to\{0,1\}$
is the standard logic and function. In particular, $n\approx\ln2\cdot w2^{w},$
$w=\log n-\log\ln n+o(1).$ Then $\mINF(\TR{}_{n})=\Theta(\log n/n)$,
and $\mu(\TR_{n})=\Theta(1).$ Apply random restriction $\cR$ that
fixes a variable with probability $1-1/w=1-\Theta(\log1/\mINF(\TR{}_{n})).$
Then for large enough $n$, 
\[
\PP[\TR_{n}|_{\cR}\equiv1]=\left(1-\left(\frac{1}{2}+\frac{1}{2w}\right)^{w}\right)^{n/w}=\Omega(1).
\]
Therefore, in this regime with constant
probability, there is no variance left under random restrictions for
the tribes function. Our bound is tight up to a $\log\log$
factor in the sense that it gives a bound up to the minimum $\rho$
where there is still some variance left after the random restriction.

\subsection{Block sensitivity is large almost everywhere}

We now move on to our second theorem, concerning block sensitivity.
The following is a nonasymptotic version of Theorem~\ref{thm:blockintro}. 
\begin{thm}
There are absolute constant $C>1.$ For any Boolean function $f:\{-1,1\}^{n}\to\{0,1\}$,
let $\tau=\mINF(f)<1/C$, $\Var[f]=2^{-o(n)}.$ Then for large enough
$n$, 
\begin{align*}
\PP_{x}\left[\mathbf{\bs}_{f}(x)\ge\frac{\Var[f]\ln1/\tau}{C\ln\ln1/\tau}\right]\ge1-\exp\left(-\Theta\left(\frac{1}{\Var[f]}\ln\ln\frac{1}{\tau}\right)\right).
\end{align*}
\end{thm}

\begin{proof}
Let 
\[
M=\left\lceil \frac{2\Var[f]\ln1/\tau}{C\ln\ln1/\tau}\right\rceil .
\]
Let $X\in\{-1,1\}^{n}$ be random. Randomly partition $[n]$ into
$M$ sets, $S_1,S_2,\ldots,S_M$, each of size $\lfloor n/M\rfloor$ with maybe a small number
of remaining indices. Note that for any $i\in[M]$, $\cR=(S_i, X)$ is a random restriction
of fixed size. Then by Theorem~\ref{thm:main}, with probability
at least $1-\exp(-\Theta(\log\log(1/\tau)/\Var[f])),$ $\Var[f|_{\cR}]>0$.
In that case, exists $T_{i}\subseteq S_{i}$, such that $f(X\oplus(-1)^{\1_{T_{i}}})\not=f(X).$
The statement thus holds by the following double-counting principle,
\begin{align*}
 & \frac{1}{2}\PP_{x}[\bs_{f}(x)<M/2]\le\PP_{\cR}[f|_{\cR}\text{ is constant}].\qedhere
\end{align*}
\end{proof}

\subsection{Decision tree complexity of random restriction to monotone functions}

We record another application of our main result regarding the decision
tree complexity of the restricted function, which is in some sense
a reverse statement to the famous H{\aa}stad's switching lemma. Let $\DT(f)$
denote the deterministic decision tree complexity of $f$.\footnote{Although
the theorem is stated with respect to the deterministic decision tree complexity,
one can replace the deterministic decision tree complexity by many other
complexity measures, for example, the randomized decision tree complexity, as
they are polynomially related for total functions.}
\begin{thm}[Decision tree complexity of random restrictions]
There are absolute constant $C>1.$ For any monotone function $f:\{-1,1\}^{n}\to\{0,1\}$,
such that 
\begin{align}
 & \log\left(\frac{1}{\mINF(f)}\right) \ge C\log\left(\frac{1}{\Var[f]}\right). \label{eq:bounds-V-vs-I-dt}
\end{align}
Let $\cR$ be a random restriction that keeps exactly $\lceil\rho n\rceil$
variables alive, where 
\[
\rho=\Omega\left(\sqrt{\frac{\log\Var[f]}{\log\mINF(f)}\log\frac{\log\mINF(f)}{\log\Var[f]}}\right).
\]
Then for large enough $n,$
\begin{equation}
\PP\left[\DT(f|_{\cR})\ge\mINF(f)^{-\Theta(\rho)}\right]\ge \frac{1}{2}. \label{eq:prob-dt-thm}
\end{equation}
\end{thm}

\begin{proof}
For simplicity, assume that $\rho n$ is a positive integer. We need
the following well-known result due to O'Donnell et al.~\cite{osss2005dt-maxinf}:
For any Boolean function $h,$
\begin{equation}
\mINF(h)\cdot\DT(h)\ge\Var[h].\label{eq:osss-ineq}
\end{equation}
Consider the  uniform  process $X(t)$. By definition,
$X((1-\rho)n)$ induces a random restriction $\cR$ that keeps exactly
$\rho n$ variables alive. By Lemma~\ref{lem:influence-remain-small-P},
\begin{align*}
\Pr_{P}\left[\max_{0\le t\le(1-\rho)n}|\partial_{i}f(X(t))|\ge\mINF(f)^{\frac{\rho}{30}}\right] & \le\mINF(f)^{\frac{\rho}{40}}+\exp(-\rho n/8).
\end{align*}
Since $f$ is monotone, the influence $\INF_{i}(f|_{\cR})=|\partial_{i}f(X((1-\rho)n))|$
for any alive coordinate $i$. The above formula then implies
\begin{equation}
\PP_{\cR}\left[\mINF(f|_{\cR})\ge\mINF(f)^{\frac{\rho}{30}}\right]\le\mINF(f)^{\frac{\rho}{40}}+\exp(-\rho n/8).\label{eq:maxInf-random-restriction}
\end{equation}
By Theorem~\ref{thm:main},
\begin{equation}
\PP\left[\Var[f|_{\cR}]\le\exp\left(-\frac{C}{\rho}\log\frac{e}{\rho}\cdot\log\frac{8e}{\Var[f]p}\right)\right]\le p.\label{eq:var-random-restriction}
\end{equation}
Set $p=1/3$. Then combining~(\ref{eq:osss-ineq})-(\ref{eq:var-random-restriction}),
\begin{align*}
\PP & \left[\DT(f)\ge\mINF(f)^{-\frac{\rho}{60}}\right]\\
 & \ge\PP\left[\DT(f)\ge\exp\left(-\frac{C}{\rho}\log\frac{e}{\rho}\cdot\log\frac{8e}{\Var[f]p}+\frac{\rho}{30}\ln\frac{1}{\mINF(f)}\right)\right]\\
 & \ge \PP\left[\Var[f|_\cR] \ge  \exp\left(-\frac{C}{\rho}\log\frac{e}{\rho}\cdot\log\frac{8e}{\Var[f]p}\right) 
  \land \left(\mINF(f|_{\cR})\le\mINF(f)^{\frac{\rho}{30}}\right)
 \right]\\
 & \ge 1-p-\mINF(f)^{\frac{\rho}{40}}-\exp(-\rho n/8)\\
 & \ge 1 - p - 2\mINF(f)^{\frac{\rho}{40}},
\end{align*}
where the first step holds for our choice of $\rho$ and~(\ref{eq:bounds-V-vs-I-dt}); the final step holds by Fact~\ref{fact:mINF-bound}. Finally, by our choice of $\rho$ and the bound on $\mINF(f)$, we have $2\mINF(f)^{\frac{\rho}{40}}<1/6$. In view of~(\ref{eq:prob-dt-thm}), we have finished the proof.
\end{proof}

\section{Random Restrictions and Hypercontractivity\label{sec:hypercontractive} }

In this section, we consider the continuous random process revealing information
about the inputs $X\in\{-1,1\}^{n}$ gradually in a bit by bit manner.
We establish a hypercontractivity theorem for this ``operator,''
and then use the new hypercontractivity theorem to show that the first-order
Fourier coefficients remain small under random restriction given that
the original function has small individual influences.

\subsection{A martingale setup for random restrictions}

Consider the following random process. Let $x\in\DC$ be a uniformly
random element. Let $(\tau_{i})_{i\in[n]}$ be random variables uniformly
distributed in the interval $[0,1]$. $\tau$ induces a permutation
on $[n]$. This is essentially the only relevant information. For
technical reasons we prefer this continuous description in this section. Define
$S(t)=\{i:\tau_{i}\le t\}$, and define process $X(t)\in\CC$ as follows
\[
X_{i}(t)=\begin{cases}
0 & \tau_{i}>t,\\
x_{i} & \tau_{i}\le t.
\end{cases}
\]
In another word, a random $\pm1$ variable is revealed with probability
$t$ at time $t$. This random process induces a random restriction
$\cR(t)=(S(t),Y)$ of function $f$, that all the variables in $S(t)$
is set according to $Y$ while the other variables are left alive.
Below, we collect some properties of a function $f$ with respect
to the above process. 
\begin{prop}
\label{prop:nabla-f}For any multilinear function $f:[-1,1]^{n}\to\RR$
and any $t\ge0,$ 
\begin{enumerate}
\item \label{enu:nabla-f-length}$\EE[|\nabla f(X(t))|^{2}]\le\|f\|_{\infty}^{2}/(1-t),$ 
\item \label{enu:del-vs-inf}$\EE[\partial_{i}f(X(t))^{2}]\le\INF_{i}[f],$
for $i=1,2,\ldots,n.$ 
\end{enumerate}
\end{prop}

\begin{proof}
\ref{enu:nabla-f-length} Note that for $i\not\in S(t),$ by definition
\[
\partial_{i}f(X(t))=\widehat{f|_{\cR(t)}}(i).
\]
Thus, by Parseval's identity, 
\begin{align*}
 & \sum_{i=1}^{n}\II\{\tau_{i}>t\}\partial_{i}f(X(t))^{2}\le\EE[(f|_{\cR(t)})^{2}]\le\|f\|_{\infty}^{2}.
\end{align*}
Since $\II\{\tau_{i}>t\}$ and $\partial_{i}f(X(t))^{2}$ are independent,
we have 
\begin{align*}
\|f\|_{\infty}^{2}\ge\EE\left[\sum_{i=1}^{n}\II\{\tau_{i}>t\}\partial_{i}f(X(t))^{2}\right] & =(1-t)\EE[|\nabla f(X(t))|^{2}].
\end{align*}

\ref{enu:del-vs-inf} By Fourier expansion of $\partial_{i}f,$ 
\begin{align*}
\EE\left[\partial_{i}f(X(t))^{2}\right] & =\EE\left[\left(\sum_{S\ni i}\hat{f}(S)\chi_{S\setminus\{i\}}(X(t))\right)^{2}\right]\\
 & =\EE\left[\sum_{S\ni i}\hat{f}(S)^{2}\II\{\tau_{j}\le t,\,\forall j\in S\setminus\{i\}\}\right]\\
 & \le\INF_{i}(f).\qedhere
\end{align*}
\end{proof}

\subsection{A hypercontractive inequality for random restrictions}

As the time $t$ increases, the process $X(t)$ reveals more information
about the location of $X(1)$. Thus, for $0\le t\le T\le1,$ we may
view $f(X(t))$ as a ``noisy'' version of $f(X(T))$. It is therefore
expected that some hypercontractive inequality holds for those two
expressions. This intuition can be made concrete by the following
theorem. 
\begin{thm}[A hypercontractive inequality]
\label{thm:HC-ineq}For any $0\le t\le T\le1,$ and any multilinear
$f:[-1,1]^{n}\to\RR,$ we have for the random process $X$ defined
in the previous section, 
\begin{equation}
\left(\EE|f(X(t))|^{2+\epsilon}\right)^{\frac{1}{2+\epsilon}}\le\left(\EE|f(X(T))|^{2}\right)^{1/2},\label{eq:HC-ineq}
\end{equation}
where 
\[
\epsilon=T-t.
\]
\end{thm}

\begin{proof}
The proof is by induction on $n$. Once we establish the base case,
the inductive step follows from a standard argument. We first show
the inductive step since the base case is more involved.

\uline{Inductive step.} Let $f(z,x)=zg(x)+h(x),$ where $x\in[-1,1]^{n}$,
and $z\in[-1,1]$. Then 
\begin{align*}
\EE[|f & (Z(t),X(t))|^{2+\epsilon}]^{\frac{1}{2+\epsilon}}\\
 & =\left(\EE_{X(t)}\left[\EE_{Z(t)}[|Z(t)g(X(t))+h(X(t))|^{2+\epsilon}]\right]\right)^{\frac{1}{2+\epsilon}}\\
 & \stackrel{(\mathrm{i})}{\le}\left(\EE_{X(t)}\left[\EE_{Z(T)}[|Z(T)g(X(t))+h(X(t))|^{2}]^{\frac{2+\epsilon}{2}}\right]\right)^{\frac{1}{2+\epsilon}}\\
 & \stackrel{(\mathrm{ii})}{\le}\left(\EE_{Z(T)}\left[\EE_{X(t)}[|Z(T)g(X(t))+h(X(t))|^{2+\epsilon}]^{\frac{2}{2+\epsilon}}\right]\right)^{\frac{1}{2}}\\
 & \stackrel{(\mathrm{iii})}{\le}\left(\EE_{Z(T)}\left[\EE_{X(T)}[(Z(T)g(X(T))+h(X(T)))^{2}]\right]\right)^{\frac{1}{2}},
\end{align*}
where (i) holds because for any fixed $X(t),$ $f=z\cdot g(X(t))+h(X(t))$
is a multilinear function on $z$, thus we can apply the inductive hypothesis;
inequality (iii) is true, again because for any fixed $Z(T),$ $f$
is a multilinear function on $x$ and we apply the inductive hypothesis; (ii) follows by the Minkowski inequality,
in particular, 
\begin{align*}
\left(\EE_{x}\left[\EE_{z}[f(z,x)^{2}]^{\frac{2+\epsilon}{2}}\right]\right)^{\frac{2}{2+\epsilon}} & \le\EE_{z}\left[\EE_{x}[|f(z,x)|^{2+\epsilon}]^{\frac{2}{2+\epsilon}}\right].
\end{align*}

\uline{Base case}. For the base case we consider two scenarios
separately. (i) $f$ is nonnegative (or, nonpositive) function. Let
$f:[-1,1]\to[0,\infty)$, say $f=ax+b$. It suffices to consider the
special case when $f=ax+1$ for some $0<a<1$ after normalization.
The reason is as follows: Since $f$ is nonnegative, $0\le|a|\le b.$
Thus, we can assume $a\ge0,$ this assumption does not change $\EE[|f(X(t))|^{p}]$.
For $b=0,$ there is nothing to prove. So we can assume $b=1$ by
normalization. For $a=0,$ $f$ is constant function. The statement
is clearly true. Finally, for the case $a=1,$ it follows from continuity.
After the above simplification, we make the actual analysis. 
\begin{align*}
\EE[(aX(t)+1)^{2+\epsilon}] & =(1-t)+\frac{t}{2}((1+a)^{2+\epsilon}+(1-a)^{2+\epsilon})\\
 & =1-t+t\sum_{k\ge0}a^{2k}\binom{2+\epsilon}{2k},\\
 & =1+t\sum_{k>0}a^{2k}\binom{2+\epsilon}{2k},
\end{align*}
where the second step uses Taylor expansion of $(1+x)^{p}$ for $|x|<1$.
Note that for $\epsilon\in[0,1]$ and any $k\ge2,$ 
\[
\binom{2+\epsilon}{2k}\le0.
\]
Hence, 
\begin{align}
\EE[(aX(t)+1)^{2+\epsilon}] & \le1+t(1+\epsilon/2)(1+\epsilon)a^{2}.\label{eq:single-bit-2+eps}
\end{align}

On the other hand, 
\begin{align}
\EE[(aX(T)+1)^{2}]^{\frac{2+\epsilon}{2}}%
 & =(1+Ta^{2})^{\frac{2+\epsilon}{2}}\nonumber \\
 & \ge1+(1+\epsilon/2)Ta^{2},\label{eq:single-bit-2}
\end{align}
where the last step follows from Fact~\ref{fact:inequality-a}. Compare~(\ref{eq:single-bit-2+eps})
and~(\ref{eq:single-bit-2}), we get that 
\[
\EE[(aX(t)+1)^{2+\epsilon}]\le\EE[(aX(T)+1)^{2}]^{1+\epsilon/2}
\]
as long as 
\begin{equation}
\epsilon\le\frac{T-t}{t}\label{eq:eps-hypercontractivity-gain}
\end{equation}
and $\epsilon\in[0,1].$

(ii) $f$ takes both positive and negative values, say $f=ax+b$.
This time it suffices to consider the special case when $f=x+b,$
for $0<b<1$. Since if $a$ is not 1, we can consider the function
$f/a.$ In addition, changing $b$ to $|b|$ does not affect $\EE[|f|^{p}]$.
Then 
\begin{align}
\EE[|X(t & )+b|^{2+\epsilon}]\nonumber \\
 & =(1-t)b^{2+\epsilon}+t/2((1+b)^{2+\epsilon}+(1-b)^{2+\epsilon})\nonumber \\
 & =(1-t)b^{2+\epsilon}+t\sum_{k\ge0}b^{2k}\binom{2+\epsilon}{2k}\nonumber \\
 & \le(1-t)b^{2}+t(1+b^{2}(1+\epsilon/2)(1+\epsilon))\nonumber \\
 & =1+(1-t)(b^{2}-1)+t(1+\epsilon/2)(1+\epsilon)b^{2},\label{eq:single-bit-2+eps-general}
\end{align}
where the third step uses the facts that $\binom{2+\epsilon}{2k}\le0$
for $k\ge2$ and $\epsilon\in[0,1],$ and that $b^{x}$ is decreasing
on $x$ for $0<b<1$. On the other hand, 
\begin{align}
(\EE[(X & (T)+b)^{2}])^{\frac{2+\epsilon}{2}}\nonumber \\
 & =(T+b^{2})^{1+\epsilon/2}\nonumber \\
 & =(1+b^{2})^{1+\epsilon/2}\left(1-\frac{1-T}{1+b^{2}}\right)^{1+\epsilon/2}\nonumber \\
 & \ge(1+(1+\epsilon/2)b^{2})\left(1-(1+\epsilon/2)\frac{1-T}{1+b^{2}}\right)\nonumber \\
 & =1+(1+\epsilon/2)b^{2}-(1+\epsilon/2)(1+b^{2}+\epsilon b^{2}/2)\frac{1-T}{1+b^{2}}\nonumber \\
 & =1+(1+\epsilon/2)b^{2}-(1+\epsilon/2)(1-T)-(1+\epsilon/2)b^{2}\epsilon(1-T)/(2+2b^{2})\nonumber \\
 & \ge1+(1+\epsilon/2)b^{2}-(1+\epsilon/2)(1-T)-(1+\epsilon/2)b^{2}\epsilon(1-T)/2,\label{eq:single-bit-2-general}
\end{align}
where the third step invokes Fact~\ref{fact:inequality-a} twice.
Let $R,L$ denote~(\ref{eq:single-bit-2-general}) and~(\ref{eq:single-bit-2+eps-general}),
respectively. Further, let $B=b^{2},$ then $R-L$ is a linear function
in $B$. To verify that $R\ge L$, one only needs verify the cases
when $B=0$ and $B=1$. Recall that $\epsilon=T-t,$ therefore 
\begin{align*}
B=0:\qquad\qquad R-L & =1-(1+\epsilon/2)(1-T)-t\\
 & =\epsilon-(1-T)\epsilon/2\\
 & =\epsilon(1+T)/2\\
 & \ge0,
\end{align*}
and 
\begin{align*}
B=1:\qquad\qquad R-L & =(1+\epsilon/2)(T-(1-T)\epsilon/2-t(1+\epsilon))\\
 & =(1+\epsilon/2)(\epsilon/2+T\epsilon/2-\epsilon t)\\
 & =(1+\epsilon/2)\epsilon/2(1+T-2t)\\
 & \ge0.
\end{align*}
This concludes our proof. 
\end{proof}
\begin{rem}
One can also prove a hypercontractive inequality of the $p$-norm
vs. 2-norm for $1<p<2$. The proof is analogous. 
\end{rem}

\subsection{
\texorpdfstring{$\ell_{\infty}$-Fourier mass of $f|_{\cR(t)}$ of the first order}{maximum first-order Fourier mass of the restricted function}\label{sec:beta-uniform}}

A key quantity in our analysis is the $\ell_{\infty}$-Fourier mass
of $f|_{\cR(t)}$ of the first order. Namely, 
\begin{equation}
\beta^{*}(t)=\max_{i\not\in S(t)}|\partial_{i}f(X(t))|.\label{eq:beta}
\end{equation}
In some sense, $\beta^{*}(t)$ represents the maximal influence of
$f|_{\cR(t)}.$ In particular, for the special case when $f$ is a
monotone Boolean function, $\beta^{*}(t)$ is exactly $\mINF(f|_{\cR(t)})$.
The importance of $\beta^{*}(t)$ will become clear in later sections.
Next, we show that with high probability $\beta^{*}(t)$ remains small
for $t$ even very close to 1. In fact, what we will show is that
\[
\beta=\max_{i\in[n]}|\partial_{i}f(X(t))|
\]
remains small with high probability. In particular, we establish the
following lemma using the hypercontractive inequality from the last
section. 
\begin{lem}[``influence'' remains small under random restriction]
\label{lem:influence-remains-small}Given $f:\{-1,1\}^{n}\to[-1,1].$
For any $0\le t<1$ such that 
\begin{equation}
\frac{8}{1-t}\ln\frac{2}{1-t}\le\ln\frac{1}{\mINF(f)}.\label{eq:lem-t-bound}
\end{equation}
Then for any $\theta\in(0,1)$, 
\[
\PP\left[\sup_{0\le s\le t}\beta(s)\ge\theta\right]\le\theta^{-3}\mINF(f)^{\frac{1-t}{8}}.
\]
\end{lem}

\begin{proof}
Take $T=(1+t)/2$ and let 
\begin{equation}
\epsilon=T-t.
\end{equation}
Then 
\begin{align}
\PP\left[\sup_{0\le s\le t}\beta(s)\ge\theta\right] & \le\PP\left[\sup_{0\le s\le t}\sum_{i=1}^{n}|\partial_{i}f(X(s))|^{2+\epsilon}\ge\theta^{2+\epsilon}\right]\nonumber \\
 & \stackrel{(\mathrm{i})}{\le}\theta^{-2-\epsilon}\sum_{i}\EE[|\partial_{i}f(X(t))|^{2+\epsilon}]\nonumber \\
 & \stackrel{(\mathrm{ii})}{\le}\theta^{-2-\epsilon}\sum_{i}(\EE[\partial_{i}f(X(T)){}^{2}])^{1+\epsilon/2}\nonumber \\
 & \stackrel{(\mathrm{iii})}{\le}\theta^{-2-\epsilon}\sum_{i}\INF_{i}(f)^{\epsilon/2}\EE[\partial_{i}f(X(T))^{2}]\nonumber \\
 & \stackrel{(\mathrm{iv})}{\le}\theta^{-2-\epsilon}\frac{\mINF(f)^{\epsilon/2}}{1-T}\nonumber \\
 & \stackrel{(\mathrm{v})}{\le}\theta^{-2-\epsilon}\mINF(f)^{\epsilon/4}%
\label{eq:beta-fine-bound}
\end{align}
where (i) is true due to Fact~\ref{fact:martingales} and Theorem~\ref{thm:Doob-ineq};
(ii) follows from Theorem~\ref{thm:HC-ineq}, (iii) follows from
Proposition~\ref{prop:nabla-f}~\ref{enu:del-vs-inf} , (iv) follows
from Proposition~\ref{prop:nabla-f}~\ref{enu:nabla-f-length} and
(v) follows by our choice of $T$, and~(\ref{eq:lem-t-bound}). 
\end{proof}

\subsection{\label{subsec:proof-influence-small}Proof of Lemma~\ref{lem:influence-remain-small-P}}

At this point, we have almost proved Lemma~\ref{lem:influence-remain-small-P} 
except that in the previous section we proved the version with the continuous
random process instead of the discrete one. Next, we show that the
continuous random process and the corresponding probability measure
$\tilde{P}$ used in Lemma~\ref{lem:influence-remains-small} is
close to the discrete uniform process generated by Procedure 1 in
Section~\ref{sec:control} with measure $P$ in the following sense. 
\begin{claim}
\label{claim:closeness-P-tildeP}Let $\cE_{t}$ be some event that
depends only on $X(t)$. Then for any  $\epsilon\in(0,1),$
\[
\Pr_{P}\left[\bigvee_{0\le t\le(1-\epsilon)n}\cE_{t}\right]\le\Pr_{\tilde{P}}\left[\bigvee_{0\le\tilde{t}\le(1-\epsilon/2)}\cE_{\tilde{t}}\right]+\exp(-\epsilon n/8).
\]
\end{claim}

\begin{proof}
We couple the two processes in the obvious way. Recall that $(\tau_{i})_{i\in[n]}$
is the random variables uniformly distributed in the interval $[0,1]$
in the continuous process. $\tau$ induces a permutation $\pi$ on
$[n]$. As time $\tilde{t}$ goes from 0 to 1 in $\tilde{P}$, whenever
a variable is set to value $v\in\{-1,1\},$ the corresponding variable
in the discrete process is also set to $v$. Recall that we denote
the set of fixed variables at time $\tilde{t}$ in $\tilde{P}$ by
$S(\tilde{t}).$ Then at time $\tilde{t}=(1-\epsilon/2)$, by Chernoff
bound, 
\[
\Pr_{\tilde{P}}[|S(\tilde{t})|<(1-\epsilon)n]\le\exp(-\epsilon n/8).
\]
Conditioning on $|S(\tilde{t})|\ge(1-\epsilon)n$, 
\[
\left\{ \bigvee_{0\le t\le(1-\epsilon)n}\cE_{t}\right\} _{P}\impliedby\left\{ \bigvee_{0\le\tilde{t}\le(1-\epsilon/2)}\cE_{\tilde{t}}\right\} _{\tilde{P}}.
\]
The claim follows.
\end{proof}
Now, Lemma~\ref{lem:influence-remain-small-P} is an immediate corollary
of Lemma~\ref{lem:influence-remains-small} and Claim~\ref{claim:closeness-P-tildeP}.
\begin{cor}[Restatement of Lemma~\ref{lem:influence-remain-small-P}]
Let $\epsilon>0$ be such that 
\[
\frac{16}{\epsilon}\ln\frac{4}{\epsilon}\le\ln\frac{1}{\mINF(f)}.
\]
Then for any $\theta\in(0,1)$,
\begin{align*}
\Pr_{P}\left[\max_{0\le t\le(1-\epsilon)n}|\partial_{i}f(X(t))|\ge\theta\right] & \le\theta^{-3}\mINF(f)^{\frac{\epsilon}{16}}+\exp(-\epsilon n/8).
\end{align*}
\end{cor}

\begin{proof}
Let $t=(1-\epsilon/2).$ The choice of $\epsilon$ guarantees that
we can apply Lemma~\ref{lem:influence-remains-small}. In view of
Claim~\ref{claim:closeness-P-tildeP}, we are done.
\end{proof}
\bibliographystyle{plain}
\bibliography{ref}

\end{document}